\documentclass[default,iicol,sn-mathphys]{sn-jnl}
\usepackage{amsmath}
\usepackage{amsthm}
\usepackage{algorithm}
\usepackage{algorithmicx}
\usepackage{algpseudocode}
\usepackage[misc]{ifsym}
\usepackage{natbib}
\usepackage{epstopdf}

\jyear{2021}%

\theoremstyle{thmstyleone}%
\newtheorem{theorem}{Theorem}
\newtheorem{condition}{Condition}[section]

\newtheorem{lemma}{Lemma}[section]
\newtheorem{notation}{Notation}[section]
\newtheorem{assumption}{Assumption}[section]

\theoremstyle{thmstyletwo}%

\theoremstyle{thmstylethree}%
\newtheorem{definition}{Definition}%

\begin{document}

\title[Article Title]{Robust Fused Lasso Penalized Huber Regression with Nonasymptotic Property and Implementation Studies}


\author[1]{\fnm{Xin} \sur{Xin}}\email{xinxin@henu.edu.cn}
\equalcont{These authors contributed equally to this work.}

\author[1]{\fnm{Boyi} \sur{Xie}}\email{byxiemath@163.com}
\equalcont{These authors contributed equally to this work.}

\author*[2]{\fnm{Yunhai}\sur{Xiao}}\email{yhxiao@henu.edu.cn}

\affil[1]{\orgdiv{School of Mathematics and Statistics}, \orgname{Henan University}, \orgaddress{\city{Kaifeng} \postcode{475000}, \country{China}}}

\affil*[2]{\orgdiv{Center for Applied Mathematics of Henan Province}, \orgname{Henan University}, \orgaddress{\city{Zhengzhou} \postcode{450046}, \country{China}}}

\abstract{
For some special data in reality, such as the genetic data, adjacent genes may have the similar function. Thus ensuring the smoothness between adjacent genes is highly necessary. But, in this case, the standard lasso penalty just doesn't seem appropriate anymore.
On the other hand, in high-dimensional statistics, some datasets are easily contaminated by outliers or contain variables with heavy-tailed distributions, which makes many conventional methods inadequate.  
To address both issues, in this paper, we propose an adaptive Huber regression for robust estimation and inference, in which, the fused lasso penalty is used to encourage the sparsity of the coefficients as well as the sparsity of their differences, i.e., local constancy of the coefficient profile.
Theoretically, we establish its nonasymptotic estimation error bounds under $\ell_2$-norm in  high-dimensional setting.
The proposed estimation method is formulated as a convex, nonsmooth and separable optimization problem, hence,  the alternating direction method of multipliers can be employed.
In the end, we perform on simulation studies and real cancer data studies, which illustrate that the proposed estimation method is more robust and predictive.
}




\maketitle
\section{Introduction}\label{sec1}
Data with ordered structures in many fields such as high-dimensional biomedical research \cite{li2018efficiently}, signal shapes \cite{li2020linearized}, and air quality analysis \cite{degras2021sparse}, have emerged a series of new challenges both computationally and statistically.
To deal with the ordered data, Tibshirani \cite{tibshirani2005sparsity} proposed a novel fused lasso regression which imposed not only on the variable coefficients, like lasso, but also on the consecutive differences of variable coefficients based on the assumed order of feature variables.
Since then, the fused lasso has attained a lot of research activities. For instance, Petersen et al. \cite{petersen2016fused} proposed a fused lasso additive model, in which each additive function is estimated to be piecewise constant.
Mao et al. \cite{mao2021robust} incorporated temporal prior information to the missing traffic data and then used the fused lasso regularization to fit the temporal correlation of traffic data.
Corsaro et al. \cite{corsaro2021fused} presented a model based on a fused lasso approach for the multi-period portfolio selection problem. Besides, the fused lasso is applied to portfolio weights, which encourages sparse solutions and becomes a penalization on the difference of wealth allocated across the assets between rebalancing dates.
Cui et al. \cite{cui2021fused} used a fused lasso framework to features reordering on the basis of their relevance with respect to the target feature. It enhanced the trade-off between the relevancy of each individual feature on the one hand and the redundancy between pairwise features on the other.
Degras et al. \cite{degras2021sparse} introduced a  fused lasso for segmenting models with multivariate time series.

Associated with the characteristics of fused lasso,  many estimation methods have been proposed, analyzed, and implemented.
For the fused lasso penalized least square model of Tibshirani et al. \cite{tibshirani2005sparsity}, if the tuning parameters are fixed, the solution will correspond to a quadratic programming problem. Hence, a two-stage dynamic algorithm named SQOPT was specifically designed. This method transformed variables to the form of reduction between positive and negative parts because of the absolute value. However, the process of finding a solution is relatively complex and the computing time is also too long, especially when the variable dimension is large. 
Li et al. \cite{li2014linearized} focused on fused lasso model and applied a well-known linearized alternating direction method of multipliers (ADMM). But for large or even huge scale datasets, they believed that the customizing advanced operator splitting type methods may be more appropriate.
Wang et al. \cite{wang2016fused} developed a new method of fused lasso with the adaptation of parameter ordering to scrutinize only adjacent-pair parameter's differences, which leads to a substantial reduction for the number of involved constraints. However, this method may be challenged by the increased computational complexity with respect to the underlying pattern of homogeneous parameters. 
Li et al. \cite{li2018efficiently} proposed a highly efficient inexact semi-smooth Newton based augmented Lagrangian method for solving challenging large-scale fused lasso problems. But the Newton method usually requires more strict conditions on functions' properties. 

Nevertheless, all these methods reviewed above ignored the case of the data being heavy-tailed. It was recently shown that, in this heavy-tailed case, using the Huber function loss instead of the least square is more appropriate. For example, Sun et al. \cite{sun2020adaptive} proposed an adaptive Huber regression with an adaptive robustification parameter which has the ability to adopt the sample size, dimension, and moments of the random noise. We note that Huber loss function with a robust parameter is easily computed and its asymptotic properties have been well studied. For instance, Huang \& Wu \cite{huang2021robust} adopted a pairwise Huber loss and then applied it in the situation when the noise only satisfies a weak moment condition.
In their work, a comparison theorem to characterize the gap between the excess generalization error and the prediction error was established. 
Liu et al. \cite{liu2021degrees} used a Huber loss function combining with a generalized lasso penalty to achieve robustness in estimation and variable selection. But they mainly focused on the formula of degrees of freedom that is used in information criteria for model selection.

On the numerical implementation progress, Sun et al. \cite{sun2020adaptive} solved the lasso penalized Huber regression by using the local iterative adaptive minimization algorithm of Fan et al. \cite{fan2018lamm}. This method could control the accuracy and statistical error while fitting the high-dimensional model, but it could not handle the data with more complex structures.
Meanwhile, Chen et al. \cite{chen2020low} designed an accelerated proximal gradient algorithm to solve the matrix elastic-net regularized multivariate Huber regression model. Although this model could reduce the negative effect of outliers on estimators, it was not able to resistant the outliers well. 
Hence, other modified robust loss functions deserve further investigations. 
Luo et al. \cite{luo2022distributed} proposed a robust distributed algorithm for fitting linear regressions when the data contains heavy-tailed or asymmetric errors with finite second moments. This procedure employed Barzilai-Borwein gradient descent and locally adaptive majorize-minimization, in low- and high-dimensional settings, respectively.
Ghosh et al. \cite{ghosh2016robust} proposed a super-resolution algorithm using a Huber norm-based maximum likelihood estimation by combining with an adaptive directional Huber-Markov regularization. This algorithm is simple and can obtain solutions but the wide-angle image with high resolution are required.

Inspired by the aforementioned works, in this paper, we propose a novel fused lasso penalized adaptive Huber regression model. We show that this estimation method can not only deal with heavy-tailed problem, but also can guarantee the smooth structures of the features.
A nature question is: can this model give a good estimator which has a smooth and sparse property?
To answer this question, we focus on an ADMM algorithm because of its widely applications in various fields, such as Xiao et al. \cite{xiao2013splitting}, Jiao et al. \cite{jiao2016alternating}. We should emphasize that ADMM has been illustrated numerically that it is highly efficient for minimization problems with separable structures in both objective function and constraints.
Another attractive feature is that the ADMM's convergence can be followed directly from some well-known convergence result according to the classical $2$-block semi-proximal ADMM by Fazel et al.\ cite{fazel2013hankel}. 

Based on the issues mentioned above, this paper aims to handle the smoothness and outliers in multivariate linear regression,
that is, we establish a fused lasso penalized adaptive Huber regression model. We show that this model possesses at least two advantages: (i) adjacent variables tend to be smooth, and (ii) the negative effect of outliers reduced.
Besides, an efficient convergent ADMM algorithm is proposed.
To the best of our knowledge, this is the first time to consider the case of dealing with smoothness and outliers in multivariate data set with heavy-tailed data.

The rest of the paper proceeds as follows. In Section \ref{sec2},  we quickly review the Huber loss and robustication parameter, followed by the proposal of fused lasso penalized adaptive Huber regression model. In Section \ref{sec3}, we sharply characterize the nonasymptotic performance of the proposed estimators in high dimension. We describe an implementation algorithm in Section \ref{sec4}. Subsequently, Section \ref{sec5} and Section \ref{sec6} are devoted to simulation studies and real data studies, respectively. In Section \ref{sec7}, we conclude this paper with some remarks.

\begin{notation}
For any multivariate $\boldsymbol u = (u_1, u_2, \ldots, u_n)^\top \in \mathbb{R}^n$, we let $\|\boldsymbol u\|_p := (\sum_{i=1}^{n} \mid u_i\mid ^p)^{1/p}$ be the $\ell_p$-norm, and specially, $\|\boldsymbol u\|_{\infty} = \max_{1\leq i\leq n} \mid u_i\mid$.
For any two multivariate $\boldsymbol u$ and $\boldsymbol v$, let $\left <\boldsymbol u, \boldsymbol v\right > = \boldsymbol u^\top \boldsymbol v$, and $\boldsymbol u \odot \boldsymbol v$ be the Hadamard (entry-wise) product of $\boldsymbol u$ and $\boldsymbol v$.
For two sequences of real numbers $\{a_n\}_{n\geq 1}$ and $\{b_n\}_{n\geq 1}$, we use $a_n \lesssim b_n$ to denote $a_n \leq C_n b_n$ for a constant $C_n > 0$. 
For a linear map $\mathcal{A}: \mathbb{R}^p \rightarrow \mathbb{R}^q$, $\mathcal{A}^*$ denotes its adjoint operator. Similarly, we denote the identity map by $\mathcal{I}$, or $\boldsymbol I$ in matrix case.
\end{notation}

\section{Estimation Method}\label{sec2}
In this section, we begin with the linear regression model
\begin{equation}\label{mod}
\boldsymbol y = \boldsymbol X\boldsymbol \beta + \boldsymbol \varepsilon,
\end{equation}
where 
$\boldsymbol X = (\boldsymbol x_1, \boldsymbol x_2,\ldots,\boldsymbol x_p) \in \mathbb{R}^{n\times p}$ denotes a predictor matrix,
$\boldsymbol y = (y_1,y_2,\ldots,y_n)^\top\in\mathbb{R}^{n}$ be a response vector,
$\boldsymbol \beta = (\beta_1,\beta_2,\ldots,\beta_p)^\top \in \mathbb{R}^{p}$ be an unknown parameter vector,
$\boldsymbol \varepsilon=(\varepsilon_1,\varepsilon_2,\ldots,\varepsilon_n)^\top \in \mathbb{R}^{n}$ be a random error vector with $\varepsilon_i$ being assumed independent and identically-distributed (i.i.d.) with mean $0$ and variance $\sigma_{\varepsilon}^2$,
$n$ be the sample size and $p$ be the number of predictors.
A common assumption is that the true coefficient vector $\boldsymbol \beta^*$ is sparse which guarantees the model identifiability and enhances the model fitting accuracy and interpretability.

As discussed previously,  we will use a Huber loss function for a robust estimator of  $\boldsymbol \beta$ in the model (\ref{mod}).
First of all, we review an univariate huber function for a given positive scalar $\tau > 0$, that is
\begin{equation*}
h_{\tau}(x)=\left\{
\begin{aligned}
&\frac{1}{2}x^2, &&{\text{if}~\mid x\mid\leq\tau},\\
&\tau\mid x\mid - \frac{1}{2}\tau^2, &&{\text{if}~\mid x\mid >\tau}.
\end{aligned}
\right.
\end{equation*}
The scalar $\tau$ can be viewed as a shape parameter for controlling the amount of robustness. The Huber’s criterion is similar to least square for a larger $\tau$ while it becomes more similar to LAD criterion for a smaller $\tau$. In \cite{huber1981robust}, Huber fixed $\tau= 1.345$ to get $95\%$ efficiency at $\mathcal{N}(0, 1)$. In practice, an optimal $\tau$ can be determined by cross validation or on the basis of independent validation data sets.

For model (\ref{mod}), the multivariate Huber function is defined as the mean of the Huber function defined on component wise, that is,
$$
\mathcal{L}_{\tau}(\boldsymbol x)=\frac{1}{n}\sum_{i=1}^{n}h_{\tau}(x_i).
$$
By using this Huber function, we propose the following fused lasso penalized adaptive Huber regression model
\begin{equation}\label{admod}
\min\limits_{\boldsymbol \beta \in \mathbb{R}^p} 
\mathcal{L}_{\tau}(\boldsymbol y - \boldsymbol X\boldsymbol \beta) + \lambda_1\|\boldsymbol \beta\|_1 + \lambda_2\sum_{j=2}^{p}\mid\beta_{j} - \beta_{j-1}\mid,
\end{equation}
where $\lambda_1>0$ and $\lambda_2>0$ are tuning parameters.
As usual, we call $\lambda_1\|\boldsymbol \beta\|_1 + \lambda_2\sum_{j=2}^{p}\mid\beta_{j} - \beta_{j-1}\mid$ the fused lasso penalty term. The most important merit of using fused lasso penalty term is that it has the ability to posses the smoothness property as well as variable selection. It is known that the $\ell_1$-norm is a convex relaxation of $\ell_0$-norm so that it can induce a sparse estimator. For convenience, we simplify the term $\sum_{j=2}^{p}\mid\beta_{j} - \beta_{j-1}\mid$ as $\|\boldsymbol D \boldsymbol \beta\|_1$, where
$$
\boldsymbol D=
\begin{bmatrix}
-1  &   1    & 0       & \cdots\ & 0\\
\vdots  & \vdots & \ddots  & \ddots  & \vdots\\
0   &  -1    & \cdots\ & 1       &0 \\
0   &   0    & \cdots\ & -1      & 1 \\
\end{bmatrix}
$$
is named a difference operator matrix.

As we know, there are many efficient numerical approaches that can be employed to solve (\ref{admod}) in the special case that $\tau$ is sufficiently large, i.e., the least square loss case. However,  numerical algorithms for this Huber loss in the form of (\ref{admod}) have never been investigated. Hence, developing an efficient and robust numerical algorithm for solving (\ref{admod}) becomes vitally important.

For the convenience of the later developments, we write (\ref{admod}) equivalently as the following model by using a triple of auxiliary variables, that is 
\begin{equation}\label{admod1}
	\begin{array}{rl}
\min\limits_{\boldsymbol z\in \mathbb{R}^n,\boldsymbol \alpha,\boldsymbol \beta,\boldsymbol \gamma\in \mathbb{R}^p}  & 
\mathcal{L}_{\tau}(\boldsymbol y - \boldsymbol z) + \lambda_1\|\boldsymbol \alpha\|_{1} + \lambda_2\|\boldsymbol \gamma\|_{1}\\ 
\text{s.t.}  &  \begin{pmatrix} \boldsymbol z\\ \boldsymbol \alpha\\ \boldsymbol \gamma \end{pmatrix} = 
\begin{pmatrix} \boldsymbol X\\ \boldsymbol I\\ \boldsymbol D \end{pmatrix} \boldsymbol \beta,
\end{array}
\end{equation}
where $\boldsymbol I$ is a $p$-order identity matrix.
For notational convenience, we then rewrite  (\ref{admod1}) as follows:
\begin{equation}\label{admod2}
\begin{array}{rl}
\min\limits_{\boldsymbol z\in \mathbb{R}^n,\boldsymbol \alpha,\boldsymbol \beta,\boldsymbol \gamma \in \mathbb{R}^p}
& \mathcal{L}_{\tau}(\boldsymbol y - \boldsymbol z) + \lambda_1\|\boldsymbol \alpha\|_{1} + \lambda_2\|\boldsymbol \gamma\|_1\\[2mm]
\text{s.t.}& \boldsymbol \theta = \tilde{\boldsymbol X}\boldsymbol \beta ,
\end{array} 
\end{equation}
where $\boldsymbol \theta := (\boldsymbol z~\boldsymbol \alpha~\boldsymbol \gamma)^\top$ and $\tilde{\boldsymbol X} := (\boldsymbol X~\boldsymbol I~\boldsymbol D)^\top$.
The implementation algorithm for solving (\ref{admod2}) will be described in Section \ref{sec4}.

\section{Statistical Theory}\label{sec3}

This section is devoted to the statistical theory of the estimation method (\ref{admod}).
For convenience, we define an empirical loss function   $\mathcal{L}_\tau(\boldsymbol \beta):= \frac{1}{n}\sum_{i=1}^{n}h_{\tau}(y_i - \boldsymbol x_i^\top \boldsymbol \beta)$, and then denote
$$
\hat{\boldsymbol \beta} \in \arg\min\limits_{\boldsymbol \beta \in \mathbb{R}^p} 
\mathcal{L}_{\tau}(\boldsymbol \beta) + \lambda_1\|\boldsymbol \beta\|_1 + \lambda_2\sum_{j=2}^{p}\mid\beta_{j} - \beta_{j-1}\mid.
$$
Besides, we use $\boldsymbol \beta^*$ to denote the ground truth. Let
$\mathcal{S} := supp(\boldsymbol \beta^*) \subseteq \{1, 2, \ldots, p\}$ be the true support set and let $\mid\mathcal{S}\mid = s$.

For characterizing the heavy-tailed random noise, we impose a bounded moment condition. Actually, this condition is  a relaxation of the commonly used sub-Gaussian assumption described as follows:
\begin{condition}[Bounded Moment Condition]\label{con1}
	 For $\delta > 0$, each element of the random error $\boldsymbol \varepsilon$ has a bounded (1+$\delta$)-th moment, that is,
	$$
	v_{\delta} = \max\limits_{i}\mathbb{E}(\mid\varepsilon_i\mid^{1+\delta}) < \infty .
	$$
\end{condition}

Let $\boldsymbol H_{\tau}(\boldsymbol \beta) := \nabla^2 \mathcal{L}_\tau(\boldsymbol \beta)$ be the Hessian matrix of the empirical loss function   $\mathcal{L}_\tau(\boldsymbol \beta)$. 
Let $\boldsymbol S_n = \frac 1n \sum_{i=1}^{n} \boldsymbol x_i \boldsymbol x_i^\top$ be the empirical Gram matrix and it is assumed to be nonsingular throughout this paper.
To establish the optimal statistical results in a general framework, we  introduce a localized version of the restricted eigenvalue conditions which can be considered as a modification  of Fan et al. \cite{fan2018lamm}.

\begin{definition}[Localized Restricted Eigenvalue; LRE] \label{def_LRE}
	The localized maximum and minimum eigenvalues of $\boldsymbol H_\tau$ is defined, respectively, as
	$$
	\begin{aligned}
	&\kappa_+(m,c_0,r) \\
	:=& \sup\limits_{\boldsymbol u, \boldsymbol \beta} \Big\{ \frac{\left< \boldsymbol u, \boldsymbol H_\tau(\boldsymbol \beta)\boldsymbol u \right>}{\|\boldsymbol u\|_2^2}:(\boldsymbol u, \boldsymbol \beta) \in \mathcal{C}(m,c_0,r) \Big\},
    \end{aligned}
$$
and
$$
\begin{aligned}
	&\kappa_-(m,c_0,r) \\
	:=& \inf\limits_{\boldsymbol u, \boldsymbol \beta} \Big\{ \frac{\left< \boldsymbol u, \boldsymbol H_\tau(\boldsymbol \beta)\boldsymbol u \right>}{\|\boldsymbol u\|_2^2}:(\boldsymbol u, \boldsymbol \beta) \in \mathcal{C}(m,c_0,r) \Big\},
	\end{aligned}
	$$
	where 
	$$
	\mathcal{C}(m,c_0,r) 
	:=      \left\{
	\begin{array}{c}
	(\boldsymbol u, \boldsymbol \beta) \in \mathbb{R}^p \times \mathbb{R}^p 
	: \forall J \subseteq \{ 1, \ldots, p\} \\
	 such \ that \  \mathcal S \subseteq J, \mid J\mid \leq m, \\
	\|\boldsymbol u_{J^c}\|_1 \leq c_0 \|\boldsymbol u_{J}\|_1, 
	\|\boldsymbol \beta - \boldsymbol \beta^*\|_1 \leq r
	\end{array}
	\right\}
	$$ 
	is a local $\ell_1$ cone.
\end{definition} 

\begin{condition}\label{con2}
	$\boldsymbol H_\tau(\boldsymbol \beta)$ satisfies the localized restricted eigenvalue condition $LRE(m,c_0, r)$, that is, 
	$$
	\kappa_{low} 
	\leq \kappa_-(m,c_0,r) 
	\leq \kappa_+(m,c_0,r) 
	\leq \kappa_{up}
	$$
	for some constants $\kappa_{low}>0$ and $\kappa_{up} > 0$.
\end{condition}

This condition is an unify to study generalized loss functions, whose Hessian may possibly depend on $\boldsymbol \beta$. 
Instead of $\boldsymbol H_\tau$, in the following definition, we turn to the restricted eigenvalues of $\boldsymbol S_n$ which is independent of $\boldsymbol \beta$.

\begin{definition}[Restricted Eigenvalue; RE]\label{def_RE}
    The restricted maximum and minimum eigenvalues of $\boldsymbol S_n$ are defined, respectively, as
	$$
	\rho_+(m,c_0) 
	:= \sup\limits_{\boldsymbol u} \Big\{ \frac{\left< \boldsymbol u, \boldsymbol S_n\boldsymbol u \right>}{\|\boldsymbol u\|_2^2}:\boldsymbol u \in \mathcal{C}(m,c_0) \Big\},
	$$
	and
	$$
	\rho_-(m,c_0) 
	:= \inf\limits_{\boldsymbol u} \Big\{ \frac{\left< \boldsymbol u, \boldsymbol S_n\boldsymbol u \right>}{\|\boldsymbol u\|_2^2}:\boldsymbol u \in \mathcal{C}(m,c_0) \Big\},
	$$
	where 
	\begin{align*}
	&\mathcal{C}(m,c_0) \\
	:=& \left\{\boldsymbol u \in \mathbb{R}^p :
	\begin{array}{c}
	 \forall J \subseteq \{ 1, \ldots, p\} \ such \ that\\
	  \mathcal S \subseteq J,\mid J\mid \leq m, 
	\|\boldsymbol u_{J^c}\|_1 \leq c_0 \|\boldsymbol u_{J}\|_1
	\end{array}
	\right\}.
	\end{align*}
\end{definition} 

\begin{condition}\label{con3}
	$\boldsymbol S_n$ satisfies the restricted eigenvalue condition $RE(m,c_0)$, that is, 
	$$
	\kappa_{low} \leq \rho_-(m,c_0) \leq \rho_+(m,c_0) \leq \kappa_{up}
	$$
	for some constants $\kappa_{low}>0$ and $\kappa_{up} > 0$.
\end{condition}

This condition is often used in high-dimensional nonasymptotic analysis. 
Furthermore, we can show that the localized restricted eigenvalues Condition \ref{con2} holds with high probability under the restricted eigenvalues Condition \ref{con3}. 
The result reported in the following lemma shows that a proper bound for the localized restricted eigenvalues of $\boldsymbol H_\tau$ can be obtained with a high probability under some conditions on the robustification parameter $\tau$ and sample size $n$. 

\begin{lemma}\label{lem1}
	Consider $\boldsymbol \beta \in \mathcal{C}(m,c_0, r)$ where $\mathcal{C}(m, c_0, r)$ is the local $\ell_1$ cone defined in Definition \ref{def_LRE}.
	Let $\tau \geq \max \{8r, c_1(m v_\delta)^{1/(1+\delta)}\}$ and $n \geq c_2 m^2t$ where $c_1, c_2 > 0$ are sufficient large constants only depending on $c_0$ and $\kappa_{low}$.
	Under Conditions \ref{con1} and \ref{con3}, there exist constants $\kappa_{low}$ and $\kappa_{up}$ such that the localized restricted eigenvalue of $\boldsymbol H_\tau (\boldsymbol \beta)$ satisfy
	$$
	0 < \kappa_{low}/2 
	\leq \kappa_-(m,c_0,r) 
	\leq \kappa_+(m,c_0,r) 
	\leq \kappa_{up} < \infty
	$$
	with probability at least $1 - e^{-t}$.
\end{lemma}

\begin{definition}
	(Bregman Divergence) For convex loss function $\mathcal{L}_\tau(\cdot)$, the Bregman divergence between $\hat{ \boldsymbol \beta}$ and $\boldsymbol \beta^*$ is defined as
	$$
	D_\mathcal{L}(\hat {\boldsymbol \beta}, \boldsymbol \beta^*) 
	:= \mathcal{L}_\tau (\hat {\boldsymbol \beta}) - \mathcal{L}_\tau (\boldsymbol \beta^*) - \left < \nabla \mathcal{L}_\tau(\boldsymbol \beta^*), \hat {\boldsymbol \beta} - \boldsymbol \beta^* \right > \geq 0.
	$$
	Furthermore, it can define the following symmetric Bregman divergence between $\hat {\boldsymbol \beta}$ and $\boldsymbol \beta^*$ as
	\begin{equation}\label{sBd}
	\begin{aligned}
	D_\mathcal{L}^s(\hat {\boldsymbol \beta}, \boldsymbol \beta^*) 
	:=& D_\mathcal{L}(\hat{ \boldsymbol \beta}, \boldsymbol \beta^*) + D_\mathcal{L}( \boldsymbol \beta^*, \hat {\boldsymbol \beta})\\
	=& \left < \nabla \mathcal{L}_\tau(\hat{\boldsymbol \beta}) - \nabla \mathcal{L}_\tau(\boldsymbol \beta^*), \hat {\boldsymbol \beta} - \boldsymbol \beta^* \right > \geq 0.
	\end{aligned}
	\end{equation}
\end{definition}

\begin{lemma}\label{dsBd}
	Let $\boldsymbol \beta_l := \boldsymbol \beta^* + l(\boldsymbol \beta - \boldsymbol \beta^*)$ with $l \in (0,1]$. For Huber loss function $\mathcal{L}_{\tau}(\cdot)$, it holds that
	$$
	D_\mathcal{L}^s(\boldsymbol \beta_l, \boldsymbol \beta^*) \leq l D_\mathcal{L}^s(\boldsymbol \beta, \boldsymbol \beta^*).
	$$
\end{lemma}

\begin{lemma}[Restricted Strong Convexity]\label{RSC}
	Under the condition  in Lemma \ref{lem1}, for any $(\boldsymbol u, \boldsymbol \beta) \in \mathcal{C}(m, c_0, r)$, we have 
	$$
	D_\mathcal{L}^s(\boldsymbol \beta, \boldsymbol \beta^*) \geq \frac{\kappa_{low}}{2}\|\boldsymbol \beta - \boldsymbol \beta^*\|_2^2
	$$
	with probability at least $1 - e^{-t}$.
\end{lemma}

\begin{lemma}\label{l1cone}
	($\ell_1$ cone property) Assume that $\|\nabla \mathcal{L}_\tau(\boldsymbol \beta^*)\|_\infty \leq \lambda_1/2$ and $\|\boldsymbol D\|_\infty = d > 0$. Let $\lambda_2 = b \lambda_1$ with $ b > 0$.
	Let $\hat {\boldsymbol \beta}$ be an optimal solution of (\ref{admod}).
	We have that $\hat {\boldsymbol \beta}$ falls in a local $\ell_1$cone, where
	$$
	\|(\hat {\boldsymbol \beta} - \boldsymbol \beta^*)_{\mathcal{S}^c}\|_1 \leq \frac{2bd + 3}{2bd + 1}\|(\hat {\boldsymbol \beta} - \boldsymbol \beta^*)_{\mathcal{S}}\|_1.
	$$
\end{lemma}

For a similar proof of this lemma, one may refer to Fan et al. \cite{fan2018lamm}. This lemma shows that the optimal solution $\hat {\boldsymbol \beta}$ of (\ref{admod}) must falls into a $\ell_1$cone. 

In light of above analysis, we are ready to present the main results on the regularized adaptive Huber estimator in high dimension. 
\begin{theorem}\label{th1}
	Let $\hat{\boldsymbol \beta}$ be the fused lasso regularized Huber estimator to (\ref{admod}).
	For any $t>0$ and $\tau_0 \geq \nu_{\delta}:=\min\{v_{\delta}^{1/(1+\delta)}, v_1^{1/2}\}$, let the robust and regularized parameters be
	\begin{align*}
	\tau &= \tau_0(n/t)^{\max\{ 1/(1+\delta), 1/2\}},\\  
	\lambda_1 &\geq  4\tau_0 (t/n)^{\min\{ \delta/(1+\delta), 1/2\}}.
	\end{align*}
	Assume that Condition \ref{con1}, \ref{con2} and \ref{con3} hold with $c_0 = \frac{2bd + 3}{2bd + 1}$, $b>0$, $d>0$ and $r \gtrsim \lambda_1\kappa_{low}^{-1} s$.
	Then we have 
	$$
	\|\hat{\boldsymbol \beta} - \boldsymbol \beta^*\|_2 
	\leq \lambda_1\kappa_{low}^{-1}\sqrt{s}
	$$
	with probability at least $1-(1+2p)e^{-t}$ as long as $n\geq c_3 m^2t$ for a certain large constant $c_3>0$.
\end{theorem}

Theorem \ref{th1} builds the nonasymptotics convergence rates of our proposed estimation method in the high-dimensional setting. 
Thus, the upper bound in Theorem \ref{th1} can be rewritten as 
$$
\|\hat{\boldsymbol \beta} - \boldsymbol \beta^*\|_2 \lesssim \sqrt{\frac{p_{\text {eff}}}{n_{\text {eff}}}}
$$
by setting the effective dimension as $p_{\text {eff}} := s$ and the effective sample size as $n_{\text {eff}} := (n/t)^{\min\{2\delta/(1+\delta), 1\}}$.
The effective dimension depends only on the sparsity while the effective sample size depends only on the sample size divided by the probability parameter.
We can observe that the rate of convergence is affected by the heavy-tailedness only through the effective sample size because the effective dimension keeps the same regardless of $\delta$.

\section{Numerical Algorithm}\label{sec4}
\subsection{Preliminary Results in Convex Analysis}
Now we quickly review some useful results in convex analysis \cite{rockafellar1970convex} for subsequent developments. Let $f:\mathbb{R}^p\rightarrow(-\infty,+\infty]$ be a proper closed convex function. 
The conjugate function of $f$ at $\boldsymbol y$ is defined as $f^*(\boldsymbol y) := \sup_{\boldsymbol x}\{\langle \boldsymbol x,\boldsymbol y \rangle - f(\boldsymbol x)\}$. It is well known that $f^*(\boldsymbol y)$ is convex and closed, proper if and only if $f$
is proper.
The proximal mapping of $\boldsymbol x$ associates to $f$ is defined by
$$
\mathcal{P}_{f(\cdot)}^c(\boldsymbol y) = \arg\min_{\boldsymbol x}\big\{ f(\boldsymbol x) + \frac{1}{2c}\| \boldsymbol x - \boldsymbol y\|_2^2\big\},
$$
where $c>0$ is a positive scalar.
The proximal mapping $\mathcal{P}_{f(\cdot)}^c(\boldsymbol y)$ exists and is unique for all $\boldsymbol y$ if $f$ is proper closed and convex. For any $\boldsymbol y\in\mathbb{R}^p$, the Moreau’s identity is
$$
\boldsymbol y = \mathcal{P}_{f(\cdot)}^c(\boldsymbol y) + c\mathcal{P}_{f^*(\cdot)}^{1/c}(\boldsymbol y/c),
$$
which means that the proximal mapping of  $f^*$ can be attained via computing the proximal mapping associated to $f$.
For example, the proximal mapping of an $\ell_1$-norm function is defined as
$$
\mathcal{P}_{\|\cdot\|_1}^c(\boldsymbol y)=\arg\min_{\boldsymbol x}\{\|\boldsymbol x\|_1+\frac{1}{2c}(\boldsymbol x - \boldsymbol y)^2\},
$$
which has an unique explicit solution, that is
\begin{equation}\label{norm1}
\mathcal{P}_{\|\cdot\|_1}^c(\boldsymbol y)=\text{sgn}(\boldsymbol y)\odot\max\{\mid\boldsymbol y\mid-c,0\},
\end{equation}
where `sgn' is a sign function and the symbol $\odot$ is Hadamard product of a pair of vectors.

Consider the following convex composite programming problem
\begin{equation}\label{standard}
\begin{array}{rl}
&\min\limits_{\boldsymbol x, \boldsymbol y}  f(\boldsymbol x) + g(\boldsymbol y) \\
&s.t. ~~\boldsymbol A\boldsymbol x + \boldsymbol B\boldsymbol y = \boldsymbol c,
\end{array}
\end{equation}
where $f(\cdot)$ and $g(\cdot)$ are proper closed convex functions, $\boldsymbol A\in\mathbb{R}^{l\times p}$ and $\boldsymbol B\in\mathbb{R}^{l\times q}$ are matrices, $\boldsymbol c\in\mathbb{R}^{l}$ is a given data.
The Lagrangian function of (\ref{standard}) as
\begin{equation}
\mathcal{L}(\boldsymbol x, \boldsymbol y; \boldsymbol z) = f(\boldsymbol x) + g(\boldsymbol y) - \left< \boldsymbol z, A\boldsymbol x + B\boldsymbol y - \boldsymbol c \right>,
\end{equation}
where $\boldsymbol z \in \mathbb{R}^n$ often referred to as the Lagrangian dual variable. Theoretically, the Lagrange dual function is defined as  minimizing  $\mathcal{L}(\boldsymbol x, \boldsymbol y; \boldsymbol z)$ on $(\boldsymbol x, \boldsymbol y)$, and the Lagrange dual problem of (\ref{standard}) is defined as maximizing 
the Lagrange dual function on $\boldsymbol z$, that is,
\begin{equation}\label{dualpp}
\max_{\boldsymbol z}~~-f^*(\boldsymbol A^\top \boldsymbol z) - g^*(\boldsymbol B^\top \boldsymbol z) + \left< \boldsymbol z, \boldsymbol c \right>.
\end{equation}

To ensure (\ref{standard}) along with its dual form (\ref{dualpp}) have optimal solutions, throughout this section, we make the following assumptions :
\begin{assumption}
	(i) The functions $f: \mathbb{R}^p \rightarrow \mathbb{R} \cup \{+ \infty\}$ and $g : \mathbb{R}^q \rightarrow \mathbb{R} \cup \{+ \infty\}$ in (\ref{standard}) are closed, proper, and convex;
	(ii) the Lagrangian function has a saddle point (i.e., there exist $(\boldsymbol x^*, \boldsymbol y^*, \boldsymbol z^*)$ for which $\mathcal{L}(\boldsymbol x^*, \boldsymbol y^*, \boldsymbol z) \leq \mathcal{L}(\boldsymbol x^*, \boldsymbol y^*, \boldsymbol z^*) \leq \mathcal{L}(\boldsymbol x, \boldsymbol y, \boldsymbol z^*)$ holds for all $(\boldsymbol x, \boldsymbol y , \boldsymbol z)$.		
\end{assumption}
Note that the second part of this assumption is equivalent to a strong duality, which can be ensured by constraint qualifications such as Slater’s condition \cite{rockafellar1970convex}.
For more details on constraints qualifications in optimization, we refer the readers to \cite{bonnans2013perturbation}.

\subsection{Optimality Conditions and Algorithm}
In this part, we use ADMM to solve fused lasso regularized adaptive Huber regression problem (\ref{admod2}).
Let $\boldsymbol \mu := (\boldsymbol \mu_{z}~\boldsymbol \mu_{\alpha}~\boldsymbol \mu_{\gamma})^\top \in \mathbb{R}^{n+2p-1}$ be multiplier associate to the constraint in (\ref{admod2}), then the Lagrangian function takes the following form
\begin{equation}\label{L1}
\begin{aligned}
&\mathcal{L}(\boldsymbol z, \boldsymbol \alpha, \boldsymbol \beta, \boldsymbol \gamma; \boldsymbol \mu)\\
=&\mathcal{L}_{\tau}(\boldsymbol y - \boldsymbol z) + \lambda_1\|\boldsymbol \alpha\|_1 + \lambda_2\|\boldsymbol \gamma\|_1
+ \boldsymbol \mu^\top(\boldsymbol \theta - \tilde{\boldsymbol X}\boldsymbol \beta).
\end{aligned}
\end{equation}
From optimization theory, it is known that finding a saddle point $(\bar{\boldsymbol z},\bar{\boldsymbol \alpha},\bar{\boldsymbol \beta},\bar{\boldsymbol \gamma};\bar{\boldsymbol \mu})$ of (\ref{L1})
is equivalent to finding
 $(\bar{\boldsymbol z},\bar{\boldsymbol \alpha},\bar{\boldsymbol \beta},\bar{\boldsymbol \gamma}; \bar{\boldsymbol \mu})$ such that the following Karush-Kuhn-Tucker (KKT) system is satisfied
\begin{equation}\label{kktP}
\left\{
\begin{array}{l}
-\partial \mathcal{L}_{\tau}(\boldsymbol y - \bar{\boldsymbol z}) + \bar{\boldsymbol \mu}_z \ni 0,\\[2mm]
\lambda_1 \partial \| \bar{\boldsymbol \alpha} \|_1 + \bar{\boldsymbol \mu}_{\alpha}  \ni 0,\\[2mm]
\tilde{\boldsymbol X}^\top \bar{\boldsymbol \mu} = 0,\\[2mm]
\lambda_2 \partial \| \bar{\boldsymbol \gamma} \|_1 + \bar{\boldsymbol \mu}_{\gamma} \ni 0,\\[2mm]
\bar{\boldsymbol \theta} - \tilde{\boldsymbol X}\bar{\boldsymbol \beta} = 0,
\end{array}
\right.
\end{equation}
where  $(\bar{\boldsymbol z},\bar{\boldsymbol \alpha},\bar{\boldsymbol \beta},\bar{\boldsymbol \gamma})$ is an optimal solution of (\ref{admod2}) and $\bar{\boldsymbol \mu}$ is an optimal solution of the corresponding dual problem.
We note that the KKT system plays a key role in the stopping criteria for the algorithm given below.

Let $\sigma>0$ be a penalty parameter. The augmented Lagrangian function to problem (\ref{admod2}) is given as
$$
\begin{aligned}
&\mathcal{L}_{\sigma}(\boldsymbol z, \boldsymbol \alpha, \boldsymbol \beta, \boldsymbol \gamma; \boldsymbol \mu)\\
:=& \mathcal{L}_{\tau}(\boldsymbol y - \boldsymbol z) + \lambda_1\|\boldsymbol \alpha\|_1 + \lambda_2\|\boldsymbol \gamma\|_1
+ \boldsymbol \mu^\top(\boldsymbol \theta - \tilde{\boldsymbol X}\boldsymbol \beta) \\
&+ \frac{\sigma}2\|\boldsymbol \theta - \tilde{\boldsymbol X}\boldsymbol \beta\|_2^2,
\end{aligned}
$$
or equivalently,
\begin{equation}\label{aL1}
	\begin{aligned}
		&\mathcal{L}_{\sigma}(\boldsymbol z, \boldsymbol \alpha, \boldsymbol \beta, \boldsymbol \gamma; \boldsymbol \mu)
		= \mathcal{L}_{\tau}(\boldsymbol y - \boldsymbol z) + \lambda_1\|\boldsymbol \alpha\|_1 + \lambda_2\|\boldsymbol \gamma\|_1\\
		&+ \frac{\sigma}{2}\|\boldsymbol \theta - \tilde{\boldsymbol X}\boldsymbol \beta + \boldsymbol \mu/\sigma\|_2^2.
	\end{aligned}
\end{equation}
Noting that there are four blocks involved in problem (\ref{admod2}) and each block is completely independent of each other, then the ADMM reviewed before can be used directly. For convenience, we view $\boldsymbol \beta$ as a group and $(\boldsymbol \alpha, \boldsymbol \gamma, \boldsymbol z)$ as another. Given an initial point, then the ADMM reduces the following iterative framework:
\begin{align}
\boldsymbol \beta^{k+1} 
=&\arg\min_{\boldsymbol \beta}\mathcal{L}_\sigma(\boldsymbol z^k, \boldsymbol \alpha^k, \boldsymbol \beta, \boldsymbol \gamma^k;\boldsymbol \mu^k),\label{subb}\\
\boldsymbol \theta^{k+1}:=&(\boldsymbol \alpha^{k+1}, \boldsymbol \gamma^{k+1}, \boldsymbol z^{k+1}) \nonumber\\
=& \arg\min_{\boldsymbol \alpha, \boldsymbol \gamma, \boldsymbol z}\mathcal{L}_\sigma(\boldsymbol z, \boldsymbol \alpha, \boldsymbol \beta^{k+1}, \boldsymbol \gamma;\boldsymbol \mu^k),\label{subthe}\\
\boldsymbol \mu^{k+1} 
=& \boldsymbol \mu^k + \pi\sigma(\boldsymbol \theta^{k+1} - \tilde{\boldsymbol X}\boldsymbol \beta^{k+1}),\label{submul}
\end{align}
where $\pi \in (0, (1+\sqrt{5})/2)$ is a step length. Clearly, the main computational cost lies in the subproblems with respect to variables $\boldsymbol z$, $\boldsymbol \alpha$, $\boldsymbol \beta$ and $\boldsymbol \gamma$. In the following, we show that each subproblem admits closed form solutions which make this iterative framework is easily implemented. 

\subsection{Subproblems' Solving}
This part is devoted to solving the $\boldsymbol \beta$-subproblem and  the  $(\boldsymbol \alpha, \boldsymbol \gamma, \boldsymbol z)$-subproblem involved in (\ref{subb}) and (\ref{subthe}), respectively. 

With fixed values of other variables, the $\boldsymbol \beta$-subprolem takes the following form
\begin{equation*}
\begin{aligned}
\boldsymbol \beta^{k+1}
=&\arg\min\limits_{\boldsymbol \beta}\mathcal{L}_\sigma(\boldsymbol z^k, \boldsymbol \alpha^k, \boldsymbol \beta, \boldsymbol \gamma^k; \boldsymbol \mu^k)\\
=&\arg\min\limits_{\boldsymbol \beta}
-\langle\boldsymbol \mu^k,\tilde{\boldsymbol X}\boldsymbol \beta\rangle + \frac{\sigma}{2}\|\boldsymbol \theta^k - \tilde{\boldsymbol X}\boldsymbol \beta\|_2^2.
\end{aligned}
\end{equation*}	
Noting that it is actually a quadratic programming, then, its solution can be easily derived with the following compact form
\begin{equation}\label{beta}
\boldsymbol \beta^{k+1} 
= (\tilde{\boldsymbol X}^\top \tilde{\boldsymbol X})^{-1}\tilde{\boldsymbol X}^\top(\boldsymbol \theta^k + \boldsymbol \mu^k/\sigma).
\end{equation}
For the  $(\boldsymbol \alpha, \boldsymbol \gamma, \boldsymbol z)$-subproblem, we notice that the variable  $\boldsymbol \alpha$, $\boldsymbol \gamma$, and  $\boldsymbol z$ are independent of each other, which means that finding them together is equivalent to finding them one by one with an arbitrary order.
Firstly, when $\boldsymbol z:=\boldsymbol z^k$ and $\boldsymbol \gamma:=\boldsymbol \gamma^k$ are fixed, the $\boldsymbol \alpha$-subproblem can be expressed as
\begin{equation*}
\begin{aligned}
&\boldsymbol \alpha^{k+1}\\
=&\arg\min\limits_{\boldsymbol \alpha}\mathcal{L}_\sigma(\boldsymbol z^k, \boldsymbol \alpha, \boldsymbol \beta^{k+1}, \boldsymbol \gamma^k; \boldsymbol \mu^k)\\
=&\arg\min\limits_{\boldsymbol \alpha}\Big\{\lambda_1\|\boldsymbol \alpha\|_1 + \frac{\sigma}{2}\|\boldsymbol \alpha - \boldsymbol \beta^{k+1} + \boldsymbol \mu_{\alpha}^k/\sigma \|_2^2\Big\},
\end{aligned}	
\end{equation*}
which admits closed form solutions by using (\ref{norm1}), that is,
\begin{equation}\label{alpha}
\boldsymbol \alpha^{k+1} 
= \text{sgn}(\boldsymbol \xi^k) \odot \max\{\mid\boldsymbol \xi^k\mid-\lambda_1/\sigma,0\},
\end{equation}
where $\boldsymbol \xi^k = \boldsymbol \beta^{k+1} - \boldsymbol \mu_{\alpha}^{k}/\sigma$.
Secondly,  when $\boldsymbol z:=\boldsymbol z^k$ and $\boldsymbol \alpha:=\boldsymbol \alpha^{k+1}$ are fixed, the $\boldsymbol \gamma$-subproblem takes the following form
\begin{equation*}
\begin{aligned}
&\boldsymbol \gamma^{k+1}\\
=&\arg\min\limits_{\boldsymbol \gamma}\mathcal{L}_\sigma(\boldsymbol z^k, \boldsymbol \alpha^{k+1}, \boldsymbol \beta^{k+1}, \boldsymbol \gamma; \boldsymbol \mu^k)\\
=&\arg\min\limits_{\boldsymbol \gamma}\Big\{\lambda_2\|\boldsymbol \gamma\|_1 + \frac{\sigma}{2}\|\boldsymbol \gamma - \boldsymbol D\boldsymbol \beta^{k+1} + \boldsymbol \mu_{\gamma}^k/\sigma\|_2^2\Big\},
\end{aligned}	
\end{equation*}
which also admits closed form solutions from (\ref{norm1}), that is,
\begin{equation}\label{gamma}
\boldsymbol \gamma^{k+1}
=\text{sgn}(\boldsymbol \eta^k)\odot\max\{\mid\boldsymbol \eta^k\mid-\lambda_2/\sigma,0\},
\end{equation}
where $\boldsymbol \eta^k = \boldsymbol D\boldsymbol \beta^{k+1} - \boldsymbol \mu_{\gamma}^k/\sigma$.
Thirdly, when $\boldsymbol \gamma:=\boldsymbol \gamma^{+1}$ and $\boldsymbol \alpha:=\boldsymbol \alpha^{k+1}$ are fixed, the $\boldsymbol z$-subproblem can be computed element-wise regarding to $n$ independent one-dimensional problems, that is,
\begin{equation*}
\begin{aligned}
&\boldsymbol z^{k+1}\\
=&\arg\min\limits_{\boldsymbol z}
\mathcal{L}_\sigma(\boldsymbol z, \boldsymbol \alpha^{k+1}, \boldsymbol \beta^{k+1}, \boldsymbol \gamma^{k+1}; \boldsymbol \mu^k)\\
=&\arg\min\limits_{\boldsymbol z}
\Big\{ \mathcal{L}_{\tau}(\boldsymbol y - \boldsymbol z) + \frac{\sigma}{2}\|\boldsymbol z - X\boldsymbol \beta^{k+1} + \boldsymbol \mu_{z}^k/\sigma\|_2^2\Big\}\\
=&\sum_{i=1}^{n}\arg\min\limits_{z_i}
\Big\{ \frac{1}{n}h_{\tau}(y_i - z_i) \\
&\quad+ \frac{\sigma}{2}(z_i - \boldsymbol x_i\boldsymbol \beta^{k+1} + \mu_{z_i}^k/\sigma)^2\Big\}.
\end{aligned}	
\end{equation*}
We observe that solving each $z_i$-subproblem
\begin{equation}\label{z1}
\arg\min\limits_{z_i}\frac{1}{n}h_{\tau}(y_i - z_i) + \frac{\sigma}{2}(z_i - \boldsymbol x_i\boldsymbol \beta^{k+1} + \mu_{z_i}^k/\sigma)^2,
\end{equation}
can be divided into the following two cases:\\
{\tt Case 1:} In the case of $\mid y_i - z_i\mid \leq \tau$, (\ref{z1}) reduces to 
\begin{equation*}
\arg\min\limits_{z_i}\frac{1}{2n}(y_i - z_i)^2 + \frac{\sigma}{2}(z_i - \boldsymbol x_i\boldsymbol \beta^{k+1} + \mu_{z_i}^k/\sigma)^2.
\end{equation*}
Then, its solution takes the following form
\begin{equation}\label{z2}
z_i^{k+1} = (y_i + n \sigma \boldsymbol x_i\boldsymbol \beta^{k+1} - n\mu_{z_i}^k)/(n \sigma + 1).
\end{equation}
{\tt Case 2:} In the case of $\mid y_i - z_i\mid > \tau$, (\ref{z1}) reduces to 
\begin{equation*}
\arg\min\limits_{z_i}\frac{\tau}{n}\mid y_i - z_i\mid + \frac{\sigma}{2}(z_i - \boldsymbol x_i\boldsymbol \beta^{k+1} + \mu_{z_i}^k/\sigma)^2.
\end{equation*}
Let $h_i = y_i - z_i$, then it becomes
\begin{equation*}
\arg\min\limits_{z_i}\frac{\tau}{n}\mid h_i\mid + \frac{\sigma}{2}(h_i - y_i + \boldsymbol x_i\boldsymbol \beta^{k+1} - \mu_{z_i}^k/\sigma)^2,
\end{equation*}
which admits closed form solutions by using (\ref{norm1}), that is,
\begin{equation*}
h_i^{k+1}=\text{sgn}(\zeta_i^k)\cdot\max\{\mid\zeta_i^k\mid-\tau/(n\sigma),0\},
\end{equation*}
where $\zeta_i^k = y_i - \boldsymbol x_i\boldsymbol \beta^{k+1} + \mu_{z_i}^k/\sigma$.
Therefore,
\begin{equation}\label{z3}
\begin{aligned}
z_i^{k+1} 
=& y_i -h_i^{k+1} \\
=& y_i - \text{sgn}(\zeta_i^k)\odot\max\{\mid\zeta_i^k\mid-\tau/(n\sigma),0\}.
\end{aligned}
\end{equation}

In light of above analysis, we are ready to state the iterative framework of ADMM, named FHADMM, for solving the fused lasso penalized adaptive Huber regression problem (\ref{admod2}) as follows.

\begin{algorithm}[H]
	\caption{ FHADMM }\label{alg:Framwork}
	\begin{algorithmic}[1]
		\Require
		Choose the robustification parameter $\tau>0$, the regularization parameters $\lambda_1$, $\lambda_2 >0$, the penalty parameter $\sigma > 0$ and step length $\pi \in (0, (1+\sqrt{5})/2)$.
		Choose an initial point $\boldsymbol z^0 \in \mathbb{R}^n$, $\boldsymbol \alpha^0 \in \mathbb{R}^p$, $\boldsymbol \gamma^0 \in \mathbb{R}^{p-1}$ 
		and initial multipliers $\boldsymbol \mu^0 = (\boldsymbol \mu_{z}^0~~\boldsymbol \mu_{\alpha}^0~~\boldsymbol \mu_{\gamma}^0)^\top \in \mathbb{R}^{n + 2p -1}$.  
		For $k = 0, 1, 2, \dots$, do the following steps iteratively.
		\Ensure
		Iterate until a certain  `stopping criterion' is met:
		\State update $\boldsymbol \beta^{k+1}$ by (\ref{beta});
		\State update $\boldsymbol \alpha^{k+1}$ by (\ref{alpha});
		\State update $\boldsymbol \gamma^{k+1}$ by (\ref{gamma});
		\State update $\boldsymbol z^{k+1}$ by (\ref{z3});
		\State update $\boldsymbol \mu^{k+1}$ by
		\begin{equation*}
		\begin{aligned}
		&\boldsymbol \mu_{z}^{k+1} 
		= \boldsymbol \mu_{z}^k + \pi\sigma(\boldsymbol z^{k+1} - \boldsymbol X\boldsymbol \beta^{k+1});\\
		&\boldsymbol \mu_{\alpha}^{k+1}
		= \boldsymbol \mu_{\alpha}^k + \pi\sigma(\boldsymbol \alpha^{k+1} - \boldsymbol \beta^{k+1});\\
		&\boldsymbol \mu_{\gamma}^{k+1} 
		= \boldsymbol \mu_{\gamma}^k + \pi\sigma(\boldsymbol \gamma^{k+1} - \boldsymbol D\boldsymbol \beta^{k+1});
		\end{aligned}
		\end{equation*}
		\Return $\boldsymbol \beta^{k+1}$.
	\end{algorithmic}
\end{algorithm}

In the experiment part, we fix the penalty parameter as $\sigma = 0.1$ for simplicity. Certainly, other techniques can be used so as to setting this value dynamically, e.g., Boyd et al. \cite{boyd2011distributed}.
Besides, we use an unit step length, i.e., $\pi = 1$, because of its extensive using in the algorithms for statistics learning. On the other hand, the convergence in the unit steplength case can be followed directly in the existing literature, such as Fazel et al. \cite{fazel2013hankel} and Yang \& Han \cite{yang2016linear}.
To make this part is easier to follow, we state its convergence result without proof because 
FHADMM is actually a standard ADMM with blocks $\boldsymbol \beta$ and $(\boldsymbol \alpha, \boldsymbol \gamma, \boldsymbol z)$. In summary, the convergence result of FHADMM can be described as follows. 
\begin{theorem}\label{th2}
Let $\{(\boldsymbol z^k, \boldsymbol \alpha^k, \boldsymbol \beta^k, \boldsymbol \gamma^k)\}$ be generated by FHADMM. 
If $\pi \in (0,(1+\sqrt{5})/2)$, then the sequence $\{(\boldsymbol z^k, \boldsymbol \alpha^k, \boldsymbol \beta^k, \boldsymbol \gamma^k)\}$ converges to an optimal solution $(\boldsymbol{\hat z}, \boldsymbol {\hat\alpha}, \boldsymbol {\hat\beta}, \boldsymbol {\hat\gamma)}$ to (\ref{admod1}) and $\{(\boldsymbol \mu_{z}^k, \boldsymbol \mu_{\alpha}^k, \boldsymbol \mu_{\gamma}^k)\}$ converges to an optimal solution $\boldsymbol{\hat\mu}$ to the dual problem of (\ref{admod1}).
\end{theorem}
\begin{proof}
	See \cite[Theorem B1]{fazel2013hankel}. 
\end{proof}

At the end of this section, we state the stopping condition of Algorithm FHADMM.  
According to (\ref{kktP}), we use the KKT residuals to measure the quality of the derived solution, i.e.,
\begin{equation*}
Res := \max\{\phi_{\mu}, \phi_{z}, \phi_{\alpha}, \phi_{\beta}, \phi_{\gamma}\} < \text{Tol},
\end{equation*}
where
\begin{equation*}
\begin{aligned}
&\phi_{\mu} := \|\boldsymbol \theta - \tilde{\boldsymbol X}\boldsymbol \beta\|,\\
&\phi_{z} := \frac{\|\boldsymbol y - \boldsymbol z - \mathcal{P}_{\mathcal{L}_{\tau}}(\boldsymbol y - \boldsymbol z - \boldsymbol \mu_z)\|}{1 + \|\boldsymbol y - \boldsymbol z\| + \|\boldsymbol \mu_z\|},\\
&\phi_{\alpha} := \frac{\|\boldsymbol \alpha - \mathcal{P}_{\| \cdot \|_1}(\boldsymbol \alpha - \boldsymbol \mu_\alpha/\lambda_1)\|}{1 + \|\boldsymbol \alpha\| + \|\boldsymbol \mu_\alpha/\lambda_1\|},\\
&\phi_{\beta} := \|\tilde{\boldsymbol X}\boldsymbol \mu\|,\\
&\phi_{\gamma} := \frac{\|\boldsymbol \gamma - \mathcal{P}_{\| \cdot \|_1}(\boldsymbol \gamma - \boldsymbol \mu_\gamma/\lambda_2)\|}{1 + \|\boldsymbol \gamma\| + \|\boldsymbol \mu_\gamma/\lambda_2\|},
\end{aligned}
\end{equation*}
and `$\text{Tol}$' is a small error tolerance. 
In our experiments, we set $\text{Tol} = 10^{-3}$ which is illustrated to be enough to derive better quality estimations in experimental preparations. Besides,
if this stopping criterion can not meet within $2000$ number of iterations, we also force the iterative process of FHADMM terminate. In this case, we say this algorithm fails to the corresponding problem.

\section{Numerical Studies on Synthetic Data}\label{sec5}
\setcounter{equation}{0}
In this section, we present some numerical studies by using a typical synthetic data to evaluate the performance of our proposal model in both low and high dimensions. 
All runs are performed on a laptop with Intel(R) Core(TM) i$7-9750$ CPU ($2.59$ GHz) and $8$ GB RAM.

For all of our numerical studies, $n$ denotes training set size, $n_{test}$ denotes test set size and $p$ denotes the number of features, or problem's dimension.
Each row of $X$ is generated by a normal distribution with mean zero and covariance matrix $\boldsymbol \Sigma$, where $\Sigma_{ij} = 0.5^{\mid i-j\mid}$ for $1 \leq i,j \leq p$. 
We simulate data with sparse and smooth vector $\boldsymbol \beta^* := (\boldsymbol \beta_1^*, \boldsymbol \beta_2^*)$, where $\boldsymbol \beta_1^* := (1,1,1,1,1,1, 2, 1.5,1.5,1.5,1.5)$ and $\boldsymbol \beta_2^* := (0,0,...,0)$ is a zero vector which the number of zero element is $p-12$. In this part, we consider the case of $p = 50$, $200$, and $400$.

The response variable $\boldsymbol y$ is generated according to $\boldsymbol y = \boldsymbol X\boldsymbol \beta^* + \boldsymbol \varepsilon$. 
We consider three different distributions of random noise $\boldsymbol \varepsilon \in \mathbb{R}^n$: \\
(i) The normal distribution $\mathcal{N}(0,0.05^2)$;\\
(ii) The $t$ distribution with degrees of freedom 1.5;\\
(iii) The lognormal distribution $\log\mathcal{N}(0, 2^2)$.

We note that performing the algorithm FHADMM involves a robustification parameter $\tau$ and a couple of regularization parameters $\lambda_1$ and $\lambda_2$.
For $\tau$, its traditional choice is $1.345$. Here,  to get a better estimation, we choose it dynamically as $\tau = a\sqrt{n/log(p)}$ with $a = \{0.4, 0.45, ..., 1.45, 1.5\}$. 
There are many techniques on the choice of $\lambda_1$ and $\lambda_2$, such as \cite{wang2007tuning} and \cite{jiao2015primal}.
Here we choose them by minimizing the estimation errors and prediction errors simultaneously. More details are ignored here because they beyond the scope of this paper.

In this part, we also do comparisons with other two typical estimation methods:\\ 
(i) The efficient fused lasso algorithm (EFLA): This algorithm was proposed by Liu et al. \cite{liu2010efficient} which is used for a fused lasso penalized least squares estimation. One key building block in EFLA is the Fused Lasso Signal Approximator (FLSA). The package of this algorithm is available at the website: \url{https://github.com/jaredhuling/fusedlasso}. Here, we name it as ``EFLA".\\
(ii) The $\tt fusedlasso$ package  proposed by Tibshirani \& Taylor \cite{tibshirani2011solution}. This solver aims to a fused lasso regularized least squares model via a dual path algorithm. Here, we name it as ``FLDP". 

To evaluate the performance of each method, we mearure the accuracy by the using of the Mean Squared Errors (MSE), which is defined as the difference between the estimated regression coefficients $\hat{\boldsymbol \beta}$ and the truth $\boldsymbol \beta^*$ under an $\ell_2$-norm. 
To visibly observe the performance of each method, we report the estimated coefficient $\hat{\boldsymbol \beta}$ $vs.$ the ground truth $\boldsymbol \beta^*$ in Figure \ref{simu_N}, Figure \ref{simu_t}, and Figure \ref{simu_L} in the sense that the normal distribution, the $t$ distribution, and the lognormal distribution on the noise $\boldsymbol \varepsilon$ are considered, respectively. Moreover, in each distribution case, we also consider three different dimensions, say $p=50$, $200$, and $400$, and report the results row by row in each figure.
The methods used in each test case are EFLA, FLDP, and FHADMM, and the results according to each method are listed column by column at each figure. 

To compare each method in a relatively fair way, we run the code $200$ times and record the average mean (MSE), the average standard error of the residuals (std), and the average computing time (CPU). The detailed numerical results are reported in Table \ref{table1} in which the results for best performance are marked in bold. 
It should be noted that the first column in Table \ref{table1} denotes the type of distribution of noise, and the second column denotes the names we concerned where the subscribe represents its dimension. From this table, we can derive the following conclusions:\\ 
(i) Our proposed method FHADMM has better performance for $t$ distribution and lognormal distribution, and is competitive with EFLA and FLDP for normal distribution, which means that our method performs better if heavy tailed noise contained.
We think this is not surprising because the least square model is widely known to have strong theoretical guarantees under Gaussian noise. 
We also see that, under different distribution cases, FHADMM performs a little better when $p=400$, which indicates that our proposed method is more suitable for solving higher dimensional problems. And particularly, under the lognormal distribution, our proposed method is a winner, which indicates that in the case of the data being not following symmetrical distribution and having long tails, our method is the best choice. 
\\
(ii) From the standard error of residual `std', we see that the values derived by our estimation method are always lower, and they are more stable in the sense that they changes slightly at each test case. 
For each method, we also see that when the samples contain more outliers, the values of the standard error of residual may increase at each noise case. 
\\
(iii) When turning our attention to the computing time, we see that our method requires the least time at the most test cases. We note that FLDP is the slowest especially in the high-dimensional case. The reason lies in that a matrix computation is needed at each iteration, which may takes the main computing burden. 
In contrast, EFLA seems a little better because a Nesterov's method is employed to produce an approximation solution per-iteration. In summary, the series of experiments demonstrate that our proposed method is the faster and highly efficient to estimate the coefficient in the heavily tailed data case.

\begin{table}[h]
	\begin{center}
		\begin{minipage}{174pt}
		\caption{The MSE, std and CPU results of EFLA, FLDP, and FHADMM methods}
		\label{table1}
		\begin{tabular}{@{}llllllll@{}}
			\toprule
			&&EFLA&FLDP&FHADMM \\
			\midrule
			$\mathcal{N}$&MSE$_{50}$  &0.201  &$\boldsymbol{5.555e}$-$\boldsymbol{6}$  &0.112\\
			&std$_{50}$   &1.595&$\boldsymbol{7.798e}$-$\boldsymbol{6}$&0.046\\
			&CPU$_{50}$   &$\boldsymbol{0.019}$&1.230&0.049\\
			&MSE$_{200}$ &$\boldsymbol{1.367}$\textbf{e}-$\boldsymbol{6}$&16.901&0.917\\
			&std$_{200}$   &47.579&1.735e+2&$\boldsymbol{0.122}$\\
			&CPU$_{200}$   &$\boldsymbol{0.083}$&26.488&0.199\\
			&MSE$_{400}$ &1.334&$\boldsymbol{1.223}$\textbf{e}-$\boldsymbol{5}$&2.562\\
			&std$_{400}$   &8.038&$\boldsymbol{1.814}$\textbf{e}-$\boldsymbol{5}$&0.752\\
			&CPU$_{400}$   &0.400&98.026&$\boldsymbol{0.315}$\\
			$t$&MSE$_{50}$ &$\boldsymbol{1.195}$&57.713&7.256\\
			&std$_{50}$   &5.219e+3&4.756e+4&$\boldsymbol{3.318}$\\
			&CPU$_{50}$   &$\boldsymbol{0.009}$&1.249&0.057\\
			&MSE$_{200}$ &$\boldsymbol{0.290}$&2.838e+4&21.336\\
			&std$_{200}$   &2.443e+4&1.767e+5&$\boldsymbol{4.890}$\\
			&CPU$_{200}$   &0.265&30.312&$\boldsymbol{0.106}$\\
			&MSE$_{400}$ &4.068e+4&82.909&$\boldsymbol{29.312}$\\
			&std$_{400}$   &5.625e+5&7.336e+2&$\boldsymbol{4.764}$\\
			&CPU$_{400}$   &0.599&134.291&$\boldsymbol{0.282}$\\
			$\log\mathcal{N}$&MSE$_{50}$  &2.752e+8&6.064e+2&$\boldsymbol{8.196}$\\
			&std$_{50}$   &4.565e+11&3.713e+3&$\boldsymbol{4.228}$\\
			&CPU$_{50}$   &$\boldsymbol{0.032}$&0.808&0.062\\
			&MSE$_{200}$ &9.165e+12&3.399e+6&$\boldsymbol{22.509}$\\
			&std$_{200}$   &2.459e+13&2.086e+7&$\boldsymbol{6.103}$\\
			&CPU$_{200}$   &1.612&65.449&$\boldsymbol{0.228}$\\
			&MSE$_{400}$ &8.423e+11&5.279e+2&$\boldsymbol{31.589}$\\
			&std$_{400}$   &6.291e+12&1.288e+2&$\boldsymbol{5.733}$\\
			&CPU$_{400}$   &4.087&128.800&$\boldsymbol{0.230}$\\
			\botrule
		\end{tabular}
	\end{minipage}
	\end{center}
\end{table}

\begin{figure}[H]
	\centering
		\includegraphics[scale=0.28]{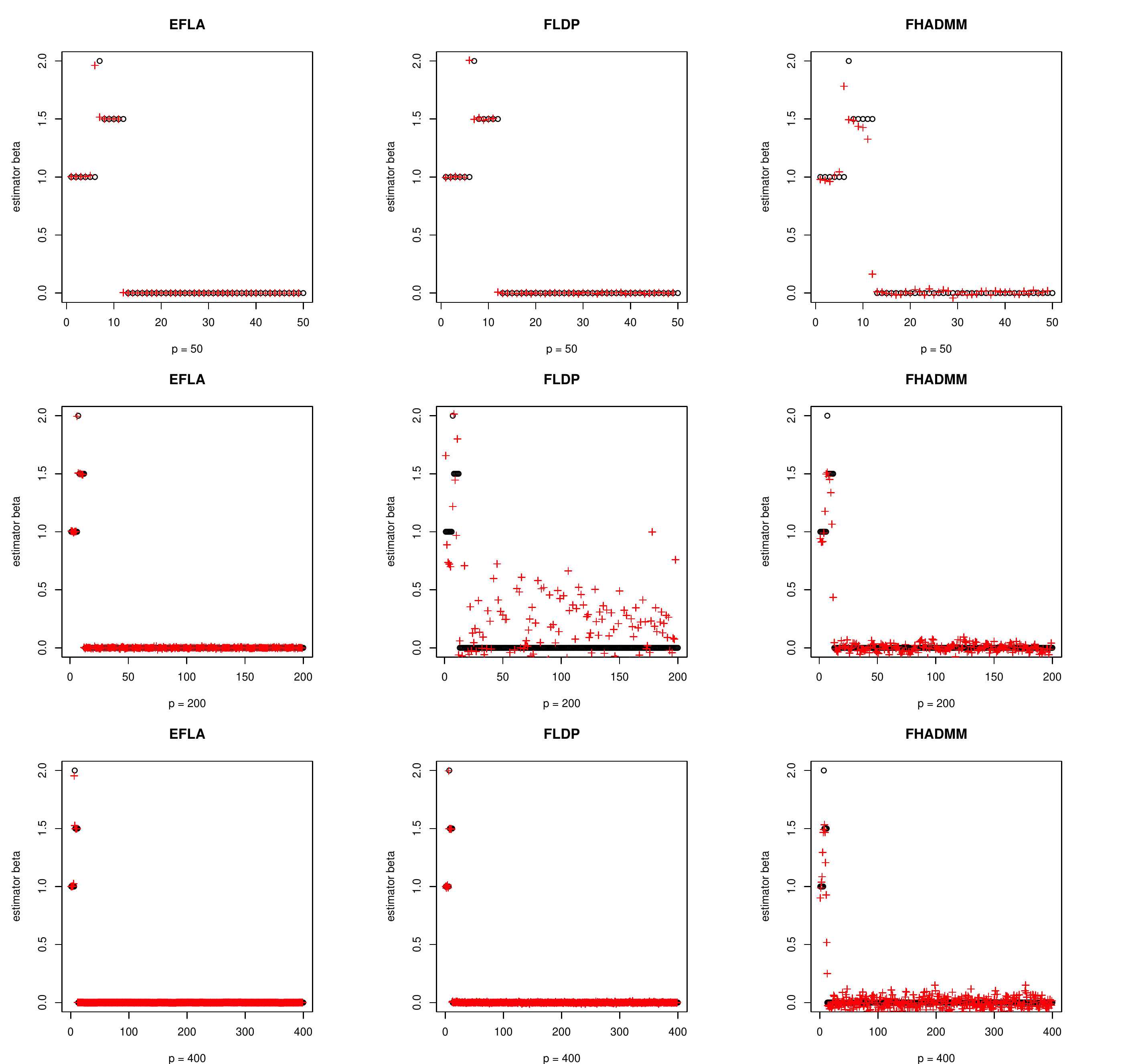}
	\caption{The ground truth `$\circ$' $vs.$ the estimator `+' under normal distribution error. The results with dimensions of $p = 50$, $200$ and $400$ are shown from top to bottom,  and with methods of EFLA, FLDP, and FHADMM are shown from left to right}
	\label{simu_N}
\end{figure}

\begin{figure}[H]
	\centering
		\includegraphics[scale=0.28]{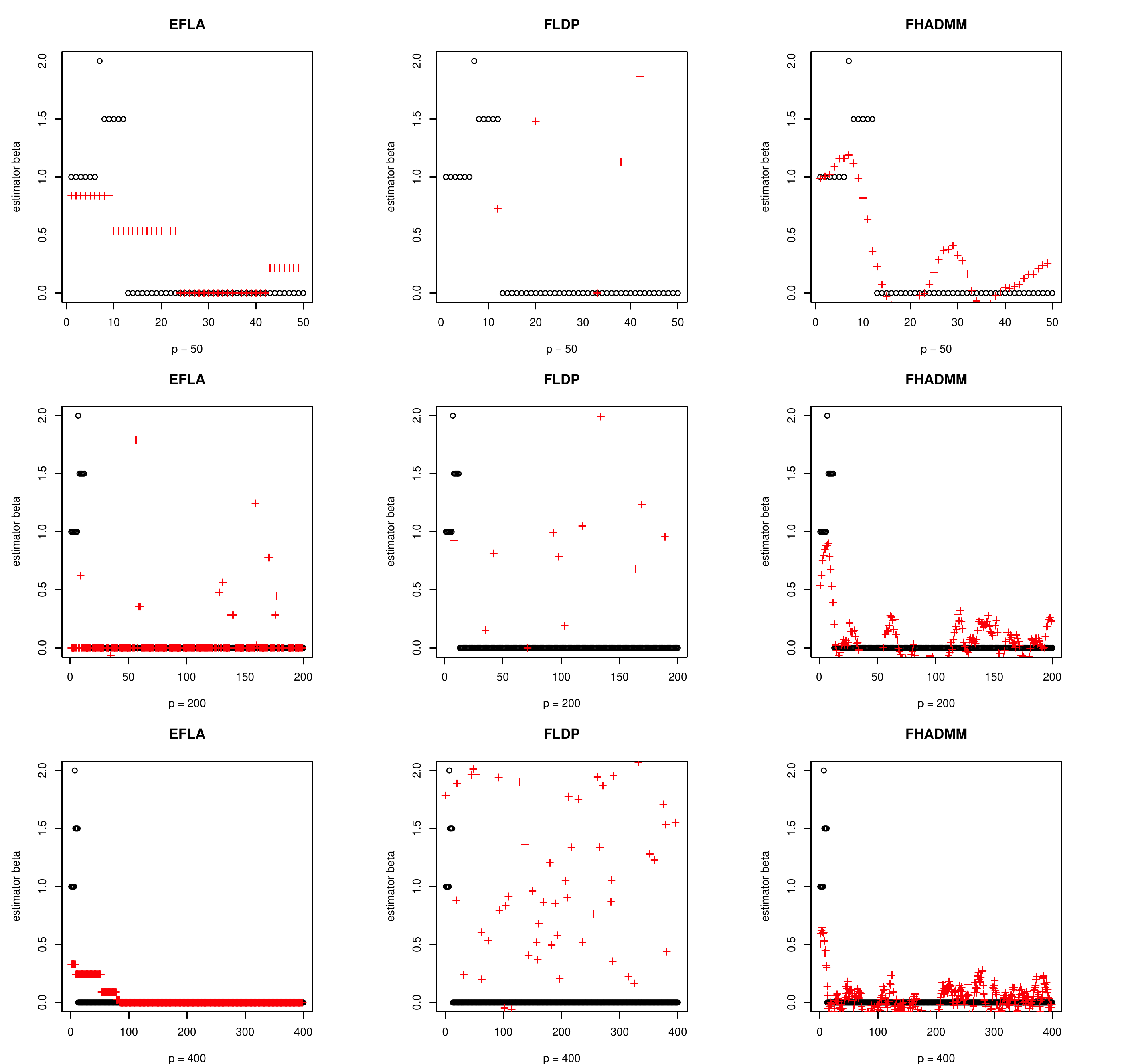}
	\caption{The ground truth `$\circ$' $vs.$ the estimator `+' under $t$ distribution error. The results with dimensions of $p = 50$, $200$ and $400$ are shown from top to bottom,  and with methods of EFLA, FLDP, and FHADMM are shown from left to right}
	\label{simu_t}
\end{figure}

\begin{figure}[H]
	\centering
		\includegraphics[scale=0.28]{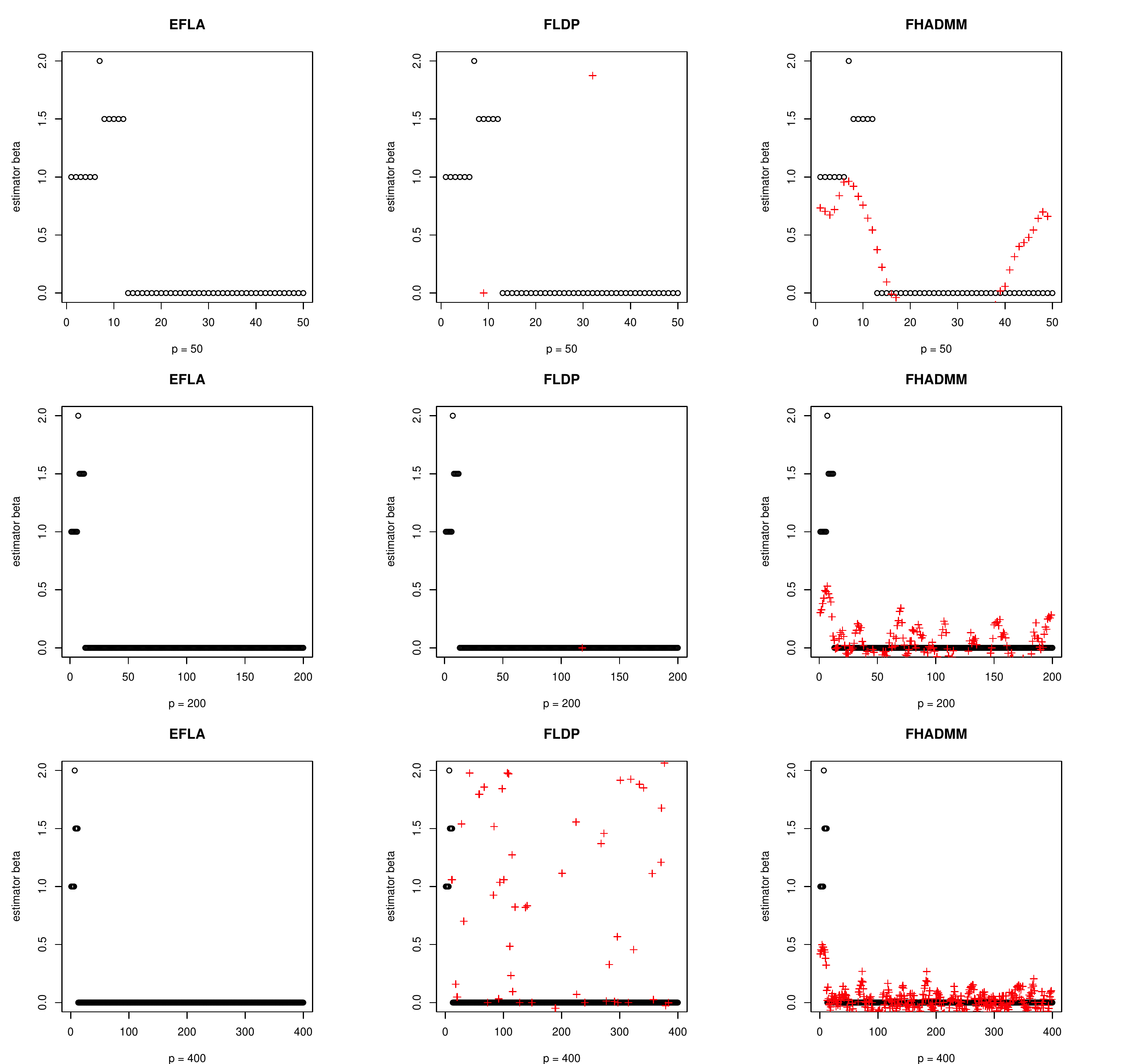}
	\caption{The ground truth `$\circ$' $vs.$ the estimator `+' under lognormal distribution error. The results with dimensions of $p = 50$, $200$ and $400$ are shown from top to bottom,  and with methods of EFLA, FLDP, and FHADMM are shown from left to right}
	\label{simu_L}
\end{figure}

\section{Numerical studies on real data}\label{sec6}
In this section, we further evaluate the effectiveness of our proposed estimation model and the progressiveness of the proposed algorithm FHADMM by using a triple of real datasets in the field of biology.
 
\subsection{Leukemia Data}
The Leukemia data was introduced by Golub et al. \cite{golub1999molecular}, which is available at the website: \url{https://hastie.su.domains/CASI\_files/DATA/leukemia.html}. 
In this data set, there are $7129$ genes and $72$ samples where $47$ in class $1$ (acute lymphocytic leukemia) and $25$ in class $2$ (acute myelogenous leukemia).
In order to explain how the gene expression level affects the biological function, we use $3707$ genes among them.

The histogram of the kurtosises for these $3707$ genes is shown in Figure \ref{kur1}. 
It shows that, there are $1980$ out of $3707$ gene expression variables have kurtosises larger than $3$, and there are $213$ out of $3707$ larger than $9$.  
In other words, there are more than $99.6\%$ of the gene expression variables have tails heavier than the normal distribution, and there are about $24.9\%$ are severely heavy-tailed with tails flatter than the $t$ distribution with degrees of freedom $5$. 
This suggests that, the genomic data can still exhibit heavy tailedness regardless of any normalization methods, see Purdom \& Holmes \cite{elizabeth2005error}.

\begin{figure}[h]
	\centering
	\includegraphics[scale=0.3]{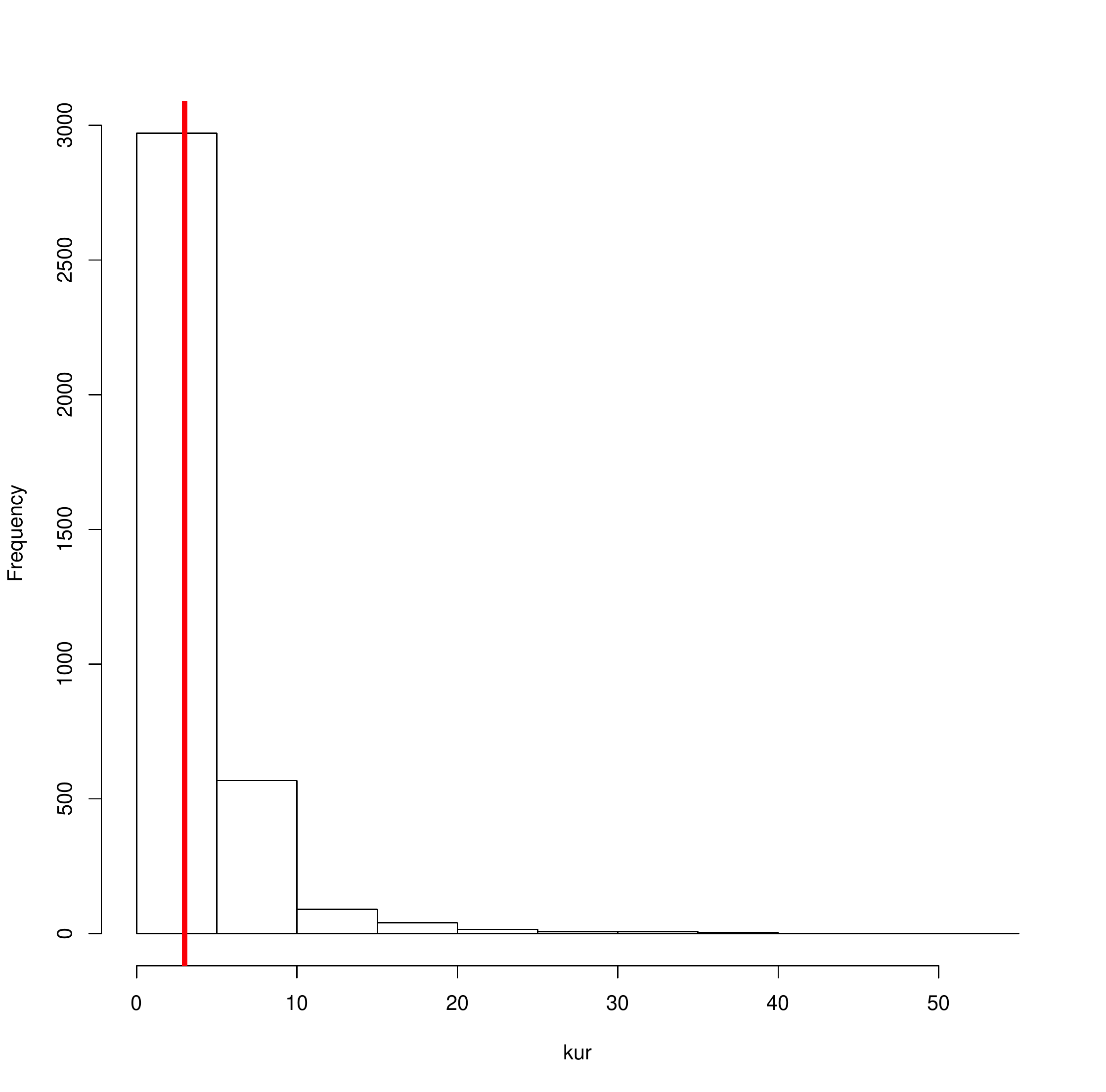}
	\caption{Histogram of Kurtosises for the leukemia gene. The thick line at $3$ is the kurtosis of a normal distribution}
	\label{kur1}
\end{figure}

We should note that there are unordered features in our data set. Therefore, we apply hierarchical clustering to order the gene expression variables so as to  exploring the property of fused lasso regularization. In this test, we divide the data set into a training data set with $n = 50$ and a test data set with $n_{test} = 22$.
To measure the predictive accuracy, we use the robust prediction loss named Mean Absolute Error (MAE) in the form of
$$
\text{MAE}(\hat{\boldsymbol \beta}) 
:= \frac{1}{n_{test}}\sum_{i=1}^{n_{test}}
\mid y_i^{test} - \left <\boldsymbol x_i^{test}, \hat{\boldsymbol \beta} \right >\mid,
$$
where $y_i^{test}$ and $\boldsymbol x_i^{test}$ , $i = 1,\ldots, n_{test}$, coming from the test data set, respectively.

\subsection{Liver Cancer Data}
The Liver cancer data was given by Wheeler et al. \cite{wheeler2017comprehensive} which is available at the website: \url{https://www.cancer.gov/about-nci/organization/ccg/research/structural-genomics/tcga}. In this data set, there are $19255$ genes and $116$ samples, i.e., $58$ in class $1$ (patient) and $58$ in class $2$ (health).
In order to find the genes with significant expression changes between their groups, we normalized the read-counts from the sequencing analysis.
By using the R package $\tt DESeq2$, we can obtain the DE result which contains log2Fold-Change and adjecent p-value.
After filtering by specific thresholds, that is, $padj < 0.01$ and $\mid log2FoldChange\mid > 1.5$, it can be got that the number of genes with obvious difference is $2597$.

The histograms about these $2597$ genes of the kurtosises is displayed at the left hand side of Figure \ref{kur2}.
It shows that, there are $709$ out of $2597$ gene expression variables have kurtosises larger than $3$, and there are $239$ out of $2597$ larger than $9$. 
Certainly, this data set also has heavy tailedness.

\begin{figure}[h]
	\centering
	\includegraphics[scale=0.3]{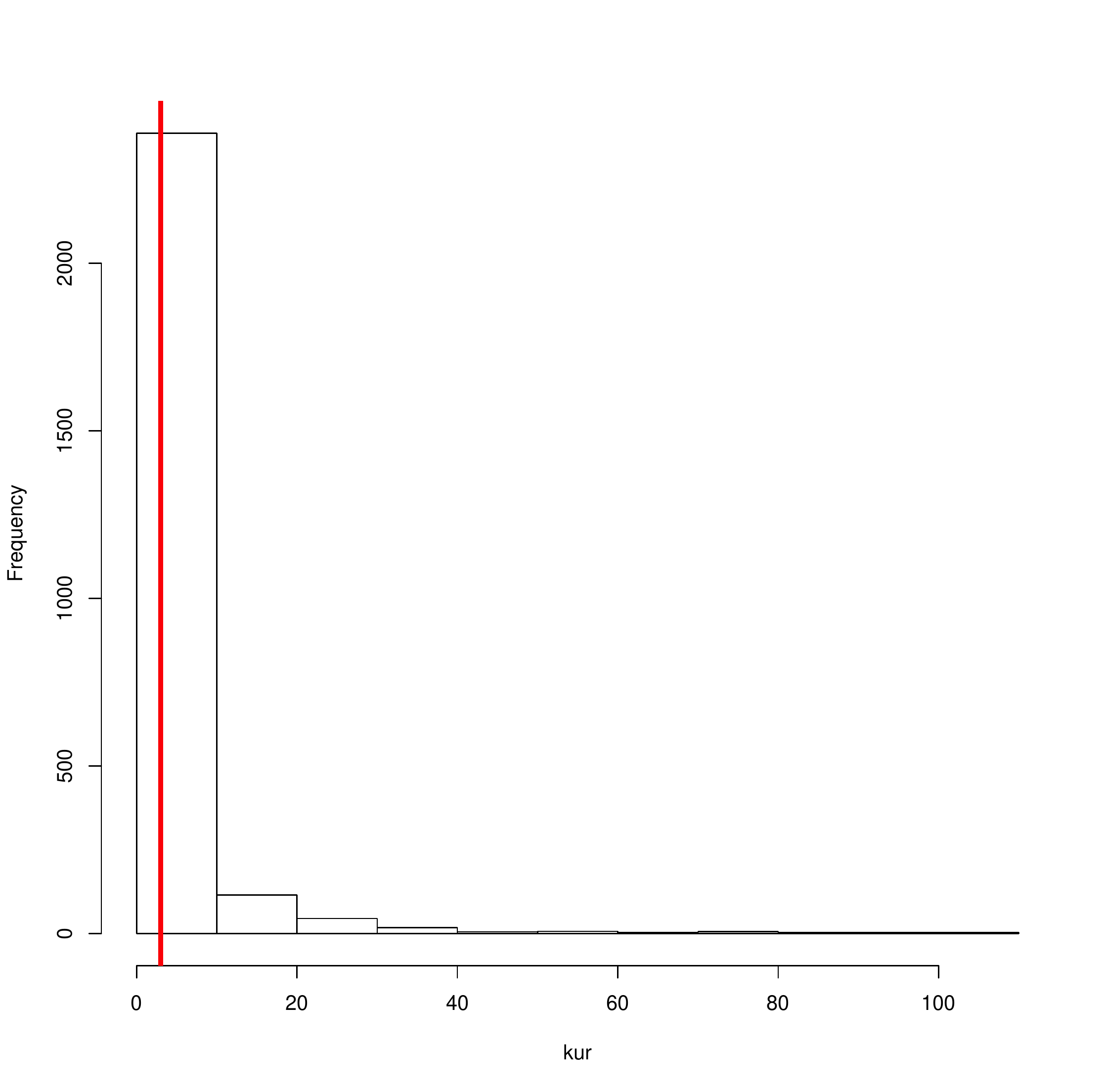}
	\caption{Histogram of Kurtosises for the liver cancer gene. The thick line at $3$ is the kurtosis of a normal distribution}
	\label{kur2}
\end{figure}

As before, we also use the hierarchical clustering to order the gene expression variables in this data set. 
In this test, this data set is divided into a training data set with $n = 81$ and a test data set with $n_{test} = 35$.
In addition, we also use the MAE to measure the algorithm's predictive performance.

\subsection{Bladder Cancer Data}
The Bladder cancer data was given in  \cite{cancer2014comprehensive} which can be download at the website \url{https://www.cancer.gov/about-nci/organization/ccg/research/structural-genomics/tcga}. 
In this data set, there are $19211$ genes and $40$ samples: $21$ in class $1$ (patients) and $19$ in class $2$ (health). 
In a similar way, we can find that the number of genes with obvious difference is $2541$.
The histograms about these $2541$ genes of the kurtosises of all expressions in the right hand side of Figure \ref{kur3}.
It can be observed that, there are $2520$ out of $2541$ gene expression variables have kurtosises larger than $3$, and there are $1494$ gene expression variables larger than $9$, which means that this data set is also heavy tailedness. 

\begin{figure}[h]
	\centering
	\includegraphics[scale=0.3]{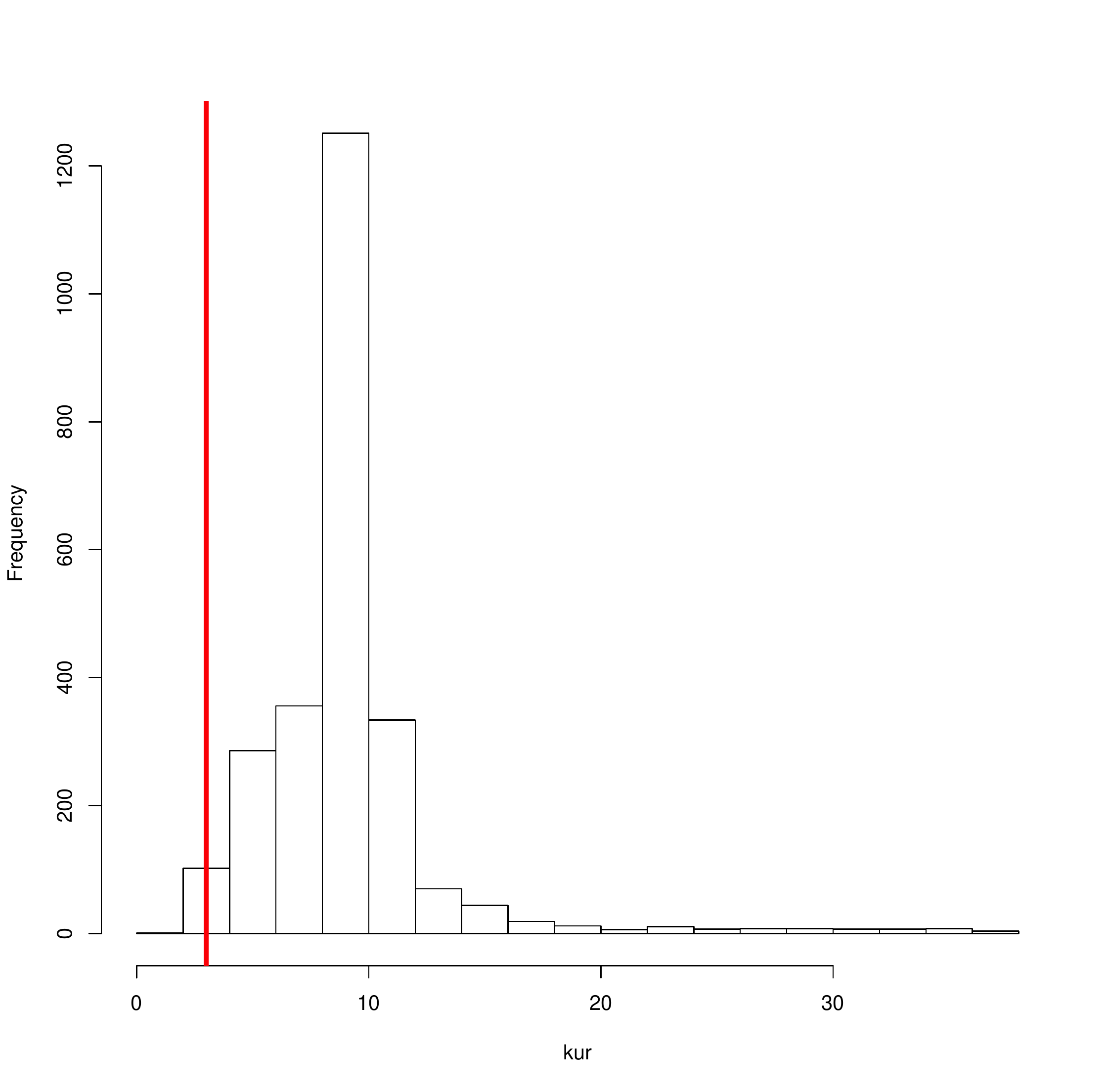}
	\caption{Histogram of Kurtosises for the bladder cancer gene. The thick line at $3$ is the kurtosis of a normal distribution}
	\label{kur3}
\end{figure}
In addition, we also apply hierarchical clustering to order the gene expression variables, and partition this data set into a training data set with $n = 28$ and a test data set with $n_{test} = 12$.

\subsection{Results' Comparisons}
We run the methods EFLA, FLDP, and FHADMM by the using of the three types of real datasets viewed above, and report the MAE values in Table \ref{table2}. From this table, we clearly observe that the MAE values derived by our proposed method FHADMM is the smallest, which is slightly smaller than the ones by FLDP and obviously smaller than the ones by EFLA. 
In summary, this simple table shows that the method FHADMM is the best but the FLDP is the worst. 

\begin{table}[h]
	\begin{center}
		\begin{minipage}{174pt}
		\caption{The MAE values of EFLA, FLDP, and FHADMM on three types of data sets}
		\label{table2}
		\begin{tabular}{@{}llll@{}}
			\toprule
			methods&Leukemia&Liver&Bladder\\
			\midrule
			EFLA&1.112&1.550&1.007\\			
			FLDP&0.944&5.029&0.960\\
			FHADMM&$\boldsymbol{0.928}$&$\boldsymbol{0.998}$&$\boldsymbol{0.779}$\\
			\botrule
		\end{tabular}
	\end{minipage}
	\end{center}
\end{table}

To evaluate the benefit of the Huber function, we compare our fused lasso penalized adaptive Huber regression model (\ref{admod})  with the fused lasso penalized least square model, i.e., the Huber function term in (\ref{admod}) is replaced by a least square. The numerical results of the methods EFLA, FLDP, and FHADMM by using the real data set `Leukemia', `Liver', and `Bladder' are report in Table \ref{table3}, \ref{table4}, and \ref{table5}, respectively. In these tables, we only display the non-zero fragments of the derived coefficients by each method based on our fused lasso penalized huber model (\ref{admod}) and the fused lasso penalized least square model. 

Comparing the values at the last column with the ones at the other columns, we see that, due to the influence of the coefficient difference constraint, it is preferable to scatter non-zero coefficients into the neighboring variables and obtain segmented smoothness solutions. 
While for the heavy tailedness, our model exhibits more robust than the model based on least square loss.

\begin{table}[h]
	\begin{center}
		\begin{minipage}{174pt}
		\caption{Parts of estimated coefficients derived by each method on Leukemia data}
		\label{table3}
		\begin{tabular}{@{}llll@{}}
			\toprule
			No. &EFLA &FLDP &FHADMM\\
			\midrule
			226  &0.00000 &0.00000 &0.00066 \\
			227  &0.00000 &0.00000 &0.00066 \\
			228  &0.00000 &0.00000 &0.00066 \\
			516  &0.00000 &0.00000 &0.00036 \\
			517  &0.00000 &0.02587 &0.00036 \\
			771  &0.00000 &0.00000 &0.00070 \\
			772  &0.00000 &0.00000 &0.00070 \\
			774  &0.00000 &0.02194 &0.00070 \\
			1808 &0.06693 &0.03654 &0.00061 \\
			1809 &0.00000 &0.01689 &0.00062 \\
			2053 &0.00000 &0.00000 &0.00059 \\
			2054 &0.00000 &0.00000 &0.00059 \\
			2109 &0.01079 &0.00000 &0.00071 \\
			2110 &0.01079 &-0.00455 &0.00071\\
			2911 &0.00000 &0.00000 &0.00096 \\
			2912 &0.00000 &0.00000 &0.00096 \\
			3048 &0.00000 &0.00000 &0.00051 \\
			3050 &0.00000 &0.00636 &0.00051 \\
			3168 &0.00000 &-0.05059 &-0.00024 \\
			3170 &0.00000 &0.00000 &-0.00024 \\
			\botrule
		\end{tabular}
	\end{minipage}
	\end{center}
\end{table}

\begin{table}[h]
	\begin{center}
		\begin{minipage}{174pt}
		\caption{Parts of estimated coefficients derived by each method on Liver data}
		\label{table4}
		\begin{tabular}{@{}llll@{}}
			\toprule
			No. &EFLA &FLDP &FHADMM \\
			\midrule
			7 &0.00000 &-2.35199 &0.00008 \\
			8 &-2.48732 &-2.38856 &0.00008\\
			108 &0.00000 &-1.26470 &0.00002 \\
			109 &0.00000 &-1.26470 &0.00002 \\
			121 &0.00000 &0.00000 &0.00009 \\
			122 &0.00000 &0.00000 &0.00009 \\
			173 &0.00000 &-3.69850 &0.00008 \\
			174 &0.00000 &0.00000 &0.00008 \\
			643 &0.00000 &0.97647 &0.00002 \\
			644 &0.00000 &0.00000 &0.00002 \\
			829 &0.00000 &0.00000 &0.00003 \\
			830 &0.00000 &0.73268 &0.00003 \\
			1090 &-0.00002 &0.00000 &0.00003\\
			1091 &0.00000 &0.00000 &0.00003 \\
			2267 &0.00000 &-3.42318 &0.00003 \\
			2268 &0.00000 &0.00000 &0.00003 \\
			2498 &0.00000 &0.00000 &0.00009 \\
			2499 &0.00000 &0.00000 &0.00009 \\
			2500 &0.00000 &0.00000 &0.00009 \\
			2501 &0.00000 &0.00000 &0.00009 \\
			\botrule
		\end{tabular}
	\end{minipage}
	\end{center}
\end{table}

\begin{table}[h]
	\begin{center}
		\begin{minipage}{174pt}
		\caption{Parts of estimated coefficients derived by each method on Bladder data}
		\label{table5}
		\begin{tabular}{@{}llll@{}}
			\toprule
			No. &EFLA &FLDP &FHADMM \\
			\midrule
			5 &0.00000 &0.03337 &0.00061\\
			6 &0.00000 &0.03337 &0.00061\\
			125 &0.00000 &0.01287 &0.00145\\
			126 &0.00000 &0.01287 &0.00145\\
			163 &0.00000 &0.00000 &0.00110\\
			164 &0.00000 &0.00000 &0.00110\\
			212 &0.00000 &0.05537 &0.00100\\
			213 &0.00000 &0.03898 &0.00100\\
			220 &0.00000 &0.00167 &0.00095\\
			221 &0.59375 &0.00167 &0.00095\\
			651 &0.00000 &0.00000 &0.00072\\
			652 &0.00000 &0.02055 &0.00072\\
			1581 &0.00000 &0.00000 &0.00075\\
			1582 &0.00000 &0.06349 &0.00075\\
			2165 &0.00000 &0.00000 &0.00055\\
			2166 &0.00000 &0.00000 &0.00055\\
			2167 &0.00000 &0.00000 &0.00055\\
			2321 &0.00000 &0.00000 &0.00056\\
			2322 &0.00000 &0.00000 &0.00056\\
			2323 &0.00000 &0.00000 &0.00056\\
			\botrule
		\end{tabular}
	\end{minipage}
	\end{center}
\end{table}

\section{Conclusion}\label{sec7}
In this paper, we focus on the fused lasso penalized adaptive Huber regression method. This method is widely used in many gene data sets because these datasets often have heavy tailedness and smoothness between adjacent genes.
In this paper, we studied the nonasympotic property of this method and gave an upper bound with a higher probability.
To implement this estimation method efficiently, we proposed an ADMM algorithm which has theoretical guarantees of global convergence in optimization literature.
In simulation studies, we showed that our estimation method is very efficient in the $t$ distribution noise case and the lognormal noise case. Especially, in a high-dimensional setting, our implemented algorithm FHADMM required less time than other state-of-the-art methods EFLA and FLDP. 

At the end of this paper, it should be list some concluding remarks. Firstly, it should be noted that there are many modified Huber loss functions which can be considered to replace the standard Huber function in (\ref{admod}), 
such as smooth non-convex Huber\cite{Zhong2012TrainingRS}, trimmed Huber\cite{chen2017robust} and so on.
This should be an interesting topic for further research. Secondly, other related penalties, such as the sparse group fused lasso for model segmentation \cite{degras2021sparse}, also deserves further investigating. Thirdly, from the optimization theory, it is general known that numerical optimization algorithms based on dual problem may poss more nice properties for algorithms' design. Hence, some higher efficient optimization algorithms based on dual formulation is worthy of developing.

\section*{Acknowledges}
This work of X. Xin is supported by Natural Science Foundation of Henan (Grant No. 202300410066). The work of Y. Xiao is supported by the National Natural Science Foundation of China (Grant No. 11971149).

\begin{appendices}
\section{}
\subsection{Proof of Lemma \ref{lem1}}
\begin{proof}
	First of all, we simply the notation $\boldsymbol H_\tau(\boldsymbol\beta) $ as $\boldsymbol H_\tau$ by ignoring the variable $\boldsymbol\beta$. 
	Without loss of generality, we normalize each column of $\boldsymbol X$ as $\|\boldsymbol x_i\|_\infty \leq 1$.
	For any $(\boldsymbol u, \boldsymbol \beta) \in \mathcal{C}(m,c_0, r)$, we have
	\begin{equation}\label{4.1}
	\begin{aligned}
	&\left <\boldsymbol u, \boldsymbol H_\tau \boldsymbol u \right > \\
	=& \boldsymbol u^\top\Big\{ \frac 1n \sum_{i=1}^{n}\boldsymbol x_i \boldsymbol x_i^\top \boldsymbol 1(\mid y_i - \left <\boldsymbol x_i, \boldsymbol \beta \right >\mid \leq \tau) \Big\}\boldsymbol u\\
	=& \|\boldsymbol S_n^{1/2}\boldsymbol u\|_2^2 - \boldsymbol u^\top\Big\{ \frac 1n \sum_{i=1}^{n}\boldsymbol x_i \boldsymbol x_i^\top \boldsymbol 1(\mid y_i - \left <\boldsymbol x_i, \boldsymbol \beta \right >\mid > \tau) \Big\}\boldsymbol u\\
	=& \|\boldsymbol S_n^{1/2}\boldsymbol u\|_2^2 \\
	&- \boldsymbol u^\top\Big\{ \frac 1n \sum_{i=1}^{n}\boldsymbol x_i \boldsymbol x_i^\top \boldsymbol 1(\mid y_i - \left <\boldsymbol x_i, \boldsymbol \beta - \boldsymbol \beta^* + \boldsymbol \beta^*\right >\mid > \tau) \Big\}\boldsymbol u\\
	\geq& \|\boldsymbol S_n^{1/2}\boldsymbol u\|_2^2 
	- \frac 1n \sum_{i=1}^{n} \left <\boldsymbol u, \boldsymbol x_i \right >^2 \boldsymbol 1(\mid \left <\boldsymbol x_i, \boldsymbol \beta - \boldsymbol \beta^*\right >\mid \geq \tau/2)\\
	&- \frac 1n \sum_{i=1}^{n} \left <\boldsymbol u, \boldsymbol x_i \right >^2 \boldsymbol 1(\mid\varepsilon_i\mid > \tau/2)\\
	\geq& \|\boldsymbol S_n^{1/2}\boldsymbol u\|_2^2 
	- \frac{2r}{\tau}\|\boldsymbol S_n^{1/2}\boldsymbol u\|_2^2\\
	&- \max\limits_{1 \leq i \leq n}\left <\boldsymbol u, \boldsymbol x_i \right >^2 \frac 1n \sum_{i=1}^{n} \boldsymbol 1(\mid\varepsilon_i\mid > \tau/2),
	\end{aligned}
	\end{equation}
	where
	\begin{align*}
	&\boldsymbol 1(\mid \left <\boldsymbol x_i, \boldsymbol \beta - \boldsymbol \beta^*\right >\mid \geq \tau/2)\\
	=& \boldsymbol 1(\frac{2}{\tau}\mid\left <\boldsymbol x_i, \boldsymbol \beta - \boldsymbol \beta^*\right >\mid \geq 1)\\
	\leq& \frac{2}{\tau}\mid\left <\boldsymbol x_i, \boldsymbol \beta - \boldsymbol \beta^*\right >\mid\\
	\leq& \frac{2}{\tau} \|\boldsymbol \beta - \boldsymbol \beta^*\|_1\\
	\leq& \frac{2r}{\tau}.
	\end{align*}
	
	Without loss of generality, we consider the special case $\|\boldsymbol u_J\|_2^2 = 1$. Moreover, for any $1\leq i\leq n$,  by using Holder's inequality, we have that
	\begin{equation}\label{4.2}
	\begin{aligned}
	\left <\boldsymbol u, \boldsymbol x_i \right > 
	\leq& \|\boldsymbol x_i\|_\infty\|\boldsymbol u\|_1 \\
	\leq& (1+c_0)\|\boldsymbol x_i\|_\infty\|\boldsymbol u_J\|_1 \\
	\leq& (1+c_0)\sqrt m\|\boldsymbol u_J\|_2
	= (1+c_0)\sqrt m.
	\end{aligned}
	\end{equation}
	
	In addition, for any $t > 0$ and $\tau > 0$, by using Markov's inequality, we have that
	$$
	\begin{aligned}
	&\mathbb{E}\Big( \frac 1n \sum_{i=1}^{n} 1(\mid\varepsilon_i\mid > \tau/2)\Big)\\
	=& \frac 1n \sum_{i=1}^{n}P(\mid\varepsilon_i\mid > \tau/2)\\
	=& \frac 1n \sum_{i=1}^{n}P(\mid\varepsilon_i\mid^{1+\delta} > (\tau/2)^{1+\delta})\\
	\leq& v_\delta (2/\tau)^{1+\delta}.
	\end{aligned}
	$$
	Furthermore, applying Hoeffding's inequality, it yields, with probability at least $1 - e^{-t}$, that
	\begin{equation}\label{4.3}
	\frac 1n \sum_{i=1}^{n} 1(\mid\varepsilon_i\mid > \tau/2) 
	\leq v_\delta (2/\tau)^{1+\delta} + \sqrt{t/(2n)}.
	\end{equation}
	
    Substituting (\ref{4.2}) and (\ref{4.3}) into (\ref{4.1}), we get
	$$
	\begin{aligned}
	&\left <\boldsymbol u, \boldsymbol H_\tau \boldsymbol u \right >\\ 
	\geq& \|\boldsymbol S_n^{1/2}\boldsymbol u\|_2^2 
	- \frac{2r}{\tau}\|\boldsymbol S_n^{1/2}\boldsymbol u\|_2^2\\
	&- (1+c_0)^2 m \Big(v_\delta (2/\tau)^{1+\delta} + \sqrt{t/(2n)} \Big).
	\end{aligned}
	$$
	Consequently, as long as $\tau \geq 8r$, we can get that the following inequality holds uniformly over $(\boldsymbol u, \boldsymbol \beta) \in \mathcal{C}(m, c_0, r)$ with probability at least $1 - e^{-t}$
	\begin{equation}
	\begin{aligned}
	&\left <\boldsymbol u, \boldsymbol H_\tau \boldsymbol u \right > \\
	\geq& \frac 34 \kappa_{low}
	- (1+c_0)^2 m \Big(v_\delta (2/\tau)^{1+\delta} + \sqrt{t/(2n)} \Big)\\
	\geq& \kappa_{low}/2,
	\end{aligned}
	\end{equation}
	whenever
	$$\tau \geq 2^{(4+\delta)/(1+\delta)}(1 + c_0)^{2/(1+\delta)}\kappa_{low}^{-1/(1+\delta)}(m v_\delta)^{1/(1+\delta)}$$
	and $n \geq 8(1 + c_0)^4 \kappa_{low}^{-2}m^2t.$
	On the other hand, it is a easy task to prove that $\left <\boldsymbol u, \boldsymbol H_\tau \boldsymbol u \right > \leq \kappa_{up}$. This completes the proof of the lemma.
	\verb| |
\end{proof}

\subsection{Proof of Lemma \ref{dsBd}}
\begin{proof}
	Let $Q(l) := D_\mathcal{L}(\boldsymbol \beta_l, \boldsymbol \beta^*) := \mathcal{L}_\tau(\boldsymbol \beta_l) - \mathcal{L}_\tau(\boldsymbol \beta^*) - \left < \nabla \mathcal{L}_\tau(\boldsymbol \beta^*), \boldsymbol \beta_l - \boldsymbol \beta^* \right >$. 
	Then we have 
	$$
	Q'(l) = \left < \nabla \mathcal{L}_\tau(\boldsymbol \beta_l) - \nabla \mathcal{L}_\tau(\boldsymbol \beta^*), \boldsymbol \beta - \boldsymbol \beta^* \right >.
	$$
	Subsequently, the symmetric Bregman divergence $D_\mathcal{L}^s(\boldsymbol \beta_l , \boldsymbol \beta^*)$ can be rewritten equivalently as
	$$
	\begin{aligned}
	D_\mathcal{L}^s(\boldsymbol \beta_l, \boldsymbol \beta^*) 
	=& \left < \nabla \mathcal{L}_\tau(\boldsymbol \beta_l) - \nabla \mathcal{L}_\tau(\boldsymbol \beta^*), l(\boldsymbol \beta - \boldsymbol \beta^*) \right >\\
	=& l Q'(l).
	\end{aligned}
	$$
	Setting $l = 1$, i.e., $D_\mathcal{L}^s(\boldsymbol \beta, \boldsymbol \beta^*) = Q'(1)$.
	
	$Q(l)$ is convex because of the convexity of $\mathcal{L}_\tau(\boldsymbol \beta_l)$ and $\left < \nabla \mathcal{L}_\tau(\boldsymbol \beta^*), \boldsymbol \beta_l - \boldsymbol \beta^* \right >$.
	Then its derivative  $Q'(l)$ is non-decreasing, which also indicates that
	$$
	D_\mathcal{L}^s(\boldsymbol \beta_l, \boldsymbol \beta^*) = lQ'(l)
	\leq
	lQ'(1) = lD_\mathcal{L}^s(\boldsymbol \beta, \boldsymbol \beta^*).
	$$
	\verb| |
\end{proof}

\subsection{Proof of Lemma \ref{RSC}}
\begin{proof}
	Recalling that
	$$
	D_\mathcal{L}^s(\boldsymbol \beta, \boldsymbol \beta^*) = \left < \nabla\mathcal{L}_\tau(\boldsymbol \beta) - \nabla\mathcal{L}_\tau(\boldsymbol \beta^*) , \boldsymbol \beta - \boldsymbol \beta^* \right >.
	$$
	Denote $\boldsymbol \Delta := \boldsymbol \beta - \boldsymbol \beta^*$. By the mean value theorem, we have
	$$
	D_\mathcal{L}^s(\boldsymbol \beta, \boldsymbol \beta^*) = \boldsymbol \Delta^\top \boldsymbol H_\tau(\tilde{\boldsymbol \beta}) \boldsymbol \Delta,
	$$
	where $\tilde{\boldsymbol \beta}$ lies between $\boldsymbol \beta^*$ and $\boldsymbol \beta$.
	Then we get
	$$
	D_\mathcal{L}^s(\boldsymbol \beta, \boldsymbol \beta^*) \geq \lambda_{\min} (\boldsymbol H_\tau(\tilde{\boldsymbol \beta})) \|\boldsymbol \beta - \boldsymbol \beta^*\|_2^2.
	$$
	It remains to show that $\lambda_{\min} (\boldsymbol H_\tau(\tilde{\boldsymbol \beta}))$ is lower bounded by a constant.
	There exists a $q \in [0,1]$ such that  $\tilde{\boldsymbol \beta} = q\boldsymbol \beta + (1 - q)\boldsymbol \beta^*$.
	Then it yields that
	$$
	\|\tilde {\boldsymbol \beta} - \boldsymbol \beta^*\|_1 
	= q\|\boldsymbol \beta - \boldsymbol \beta^*\|_1
	\leq qr,
	$$
	which means that $\tilde {\boldsymbol \beta} \in \mathcal{C}(m, c_0, r)$.
	By Lemma \ref{lem1}, we have $\lambda_{\min} (\boldsymbol H_\tau(\tilde{\boldsymbol \beta})) \geq \frac{\kappa_{low}}{2}$ with probability $1 - e^{-t}$. Hence, we have
	$$
	D_\mathcal{L}^s(\boldsymbol \beta, \boldsymbol \beta^*) \geq \frac{\kappa_{low}}{2}\|\boldsymbol \beta - \boldsymbol \beta^*\|_2^2.
	$$
	\verb| |
\end{proof}

\subsection{Proof of Lemma \ref{l1cone}}
\begin{proof}
	From the first-order optimality condition, we know that there exist $\tilde {\boldsymbol \omega}_1 \in \partial\|\hat {\boldsymbol \beta}\|_1$ and $\tilde {\boldsymbol \omega}_2 \in \partial\|\boldsymbol D \hat {\boldsymbol \beta}\|_1$ such that
	\begin{equation}\label{optcond}
	\nabla\mathcal{L}_\tau(\hat {\boldsymbol \beta}) + \lambda_1\tilde {\boldsymbol \omega}_1 + \lambda_2\tilde {\boldsymbol \omega}_2 = 0.
	\end{equation}
	From (\ref{sBd}), we have 
	\begin{equation}\label{sBd2}
	\left < \nabla \mathcal{L}_\tau(\hat{\boldsymbol \beta}) , \hat {\boldsymbol \beta} - \boldsymbol \beta^* \right > 
	\geq \left < \nabla \mathcal{L}_\tau(\boldsymbol \beta^*), \hat {\boldsymbol \beta} - \boldsymbol \beta^* \right >.
	\end{equation}
	Substituting (\ref{optcond}) into (\ref{sBd2}), we have
	\begin{equation}\label{threepart1}
	\begin{aligned}
	\left < \nabla\mathcal{L}_\tau(\boldsymbol \beta^*) , \hat {\boldsymbol \beta} - \boldsymbol \beta^* \right > 
	+& \lambda_1\left < \tilde {\boldsymbol \omega}_1, \hat{ \boldsymbol \beta} - \boldsymbol \beta^* \right >\\ 
	+& \lambda_2\left < \tilde {\boldsymbol \omega}_2, \hat {\boldsymbol \beta} - \boldsymbol \beta^* \right >\leq 0.
	\end{aligned}
	\end{equation}

    For simplicity, we use $I_1$, $I_2$, and $I_2$ to denote the first, the second, and the third term at the left-hand-side of (\ref{threepart1}). We now show their lower bound of $I_1$, $I_2$, and $I_2$ one by one. \\
	{\tt (i)} By Holder's inequality and the assumption that $\|\nabla \mathcal{L}_\tau(\boldsymbol \beta^*)\|_\infty \leq \lambda_1/2$, we have
	\begin{equation}\label{three11}
	I_1 
	\geq -\|\nabla \mathcal{L}_\tau(\boldsymbol \beta^*)\|_\infty\|\hat {\boldsymbol \beta} - \boldsymbol \beta^* \|_1
	\geq -\lambda_1/2\|\hat{ \boldsymbol \beta} - \boldsymbol \beta^* \|_1.
	\end{equation}
    {\tt(ii)} From the subgradient of $\ell_1$-norm, we have $\|\hat {\boldsymbol \beta}\|_1 = \left <\tilde {\boldsymbol \omega}_1, \hat {\boldsymbol \beta} \right >$, and that $\|\tilde {\boldsymbol \omega}_1\|_\infty \leq 1$. 
	Furthermore, it also holds that
	\begin{equation}\label{three12}
	\begin{aligned}
	I_2
	&= \lambda_1\left < \tilde{\boldsymbol \omega}_{1\mathcal{S}}, (\hat {\boldsymbol \beta} - \boldsymbol \beta^*)_{\mathcal{S}} \right > 
	+ \lambda_1\left < \tilde{\boldsymbol \omega}_{1\mathcal{S}^c}, (\hat {\boldsymbol \beta} - \boldsymbol \beta^*)_{\mathcal{S}^c} \right > \\
	&\geq -\lambda_1\|(\hat {\boldsymbol \beta} - \boldsymbol \beta^*)_{\mathcal{S}}\|_1 
	+ \lambda_1\left < \tilde{\boldsymbol \omega}_{1\mathcal{S}^c}, (\hat {\boldsymbol \beta} - \boldsymbol \beta^*)_{\mathcal{S}^c} \right > \\
	&\geq -\lambda_1\|(\hat {\boldsymbol \beta} - \boldsymbol \beta^*)_{\mathcal{S}}\|_1 
	+ \lambda_1\|(\hat {\boldsymbol \beta} - \boldsymbol \beta^*)_{\mathcal{S}^c}\|_1,  
	\end{aligned}
	\end{equation}
	where the last inequality follows from the fact that $\left < \tilde{\boldsymbol \omega}_{1\mathcal{S}^c}, \hat {\boldsymbol \beta}_{\mathcal{S}^c} \right > = \|\hat {\boldsymbol \beta}_{\mathcal{S}^c}\|_1$ and that $\boldsymbol \beta^*_{\mathcal{S}^c} = 0$.\\
	{\tt(iii)} In a similar way with {(ii)}, we get $\|\boldsymbol D\hat {\boldsymbol \beta}\|_1 = \left <\tilde {\boldsymbol \omega}_2, \hat {\boldsymbol \beta} \right >$, and that $\|\tilde {\boldsymbol \omega}_2\|_\infty \leq d$. 
	Furthermore, we have
	\begin{equation}\label{three13}
	\begin{aligned}
	I_3
	&= \lambda_2\left < \tilde{\boldsymbol \omega}_{2\mathcal{S}}, (\hat {\boldsymbol \beta} - \boldsymbol \beta^*)_{\mathcal{S}} \right > 
	+ \lambda_2\left < \tilde{\boldsymbol \omega}_{2\mathcal{S}^c}, (\hat {\boldsymbol \beta} - \boldsymbol \beta^*)_{\mathcal{S}^c} \right > \\
	&\geq -\lambda_2 d\|(\hat {\boldsymbol \beta} - \boldsymbol \beta^*)_{\mathcal{S}}\|_1 
	+ \lambda_2\left < \tilde{\boldsymbol \omega}_{2\mathcal{S}^c}, (\hat {\boldsymbol \beta} - \boldsymbol \beta^*)_{\mathcal{S}^c} \right > \\
	&\geq -\lambda_2 d\|(\hat {\boldsymbol \beta} - \boldsymbol \beta^*)_{\mathcal{S}}\|_1 
	+ \lambda_2 d\|(\hat {\boldsymbol \beta} - \boldsymbol \beta^*)_{\mathcal{S}^c}\|_1, 
	\end{aligned} 
	\end{equation}
	where the last inequality follows from the fact that $\left < \tilde{\boldsymbol \omega}_{2\mathcal{S}^c}, \boldsymbol D\hat {\boldsymbol \beta}_{\mathcal{S}^c} \right > = \|\boldsymbol D\hat {\boldsymbol \beta}_{\mathcal{S}^c}\|_1 = d\|\hat {\boldsymbol \beta}_{\mathcal{S}^c}\|_1$ and that $\boldsymbol \beta^*_{\mathcal{S}^c} = 0$.
	
	Substituting (\ref{three11}), (\ref{three12}), and (\ref{three13}) into (\ref{threepart1}), we get
	$$
	\begin{aligned}
	&-\lambda_1/2\|\hat {\boldsymbol \beta} - \boldsymbol \beta^* \|_1
	-\lambda_1\|(\hat {\boldsymbol \beta} - \boldsymbol \beta^*)_{\mathcal{S}}\|_1 
	+ \lambda_1\|(\hat {\boldsymbol \beta} - \boldsymbol \beta^*)_{\mathcal{S}^c}\|_1 \\
	&-\lambda_2 d\|(\hat {\boldsymbol \beta} - \boldsymbol \beta^*)_{\mathcal{S}}\|_1 + \lambda_2 d\|(\hat {\boldsymbol \beta} - \boldsymbol \beta^*)_{\mathcal{S}^c}\|_1 \leq 0,
	\end{aligned}
	$$
	or equivalently, 
	$$
	\|(\hat {\boldsymbol \beta} - \boldsymbol \beta^*)_{\mathcal{S}^c}\|_1 \leq \frac{2bd + 3}{2bd + 1} \|(\hat {\boldsymbol \beta} - \boldsymbol \beta^*)_{\mathcal{S}}\|_1.
	$$
	\verb| |
\end{proof}

\subsection{Proof of Theorem \ref{th1}}
\begin{proof}
	Using the first-order optimality condition as same as (\ref{optcond}), and substituting (\ref{optcond}) into (\ref{sBd}), we have 
	\begin{equation}\label{threepart2}
	\begin{aligned}
	D_\mathcal{L}^s(\hat \beta, \beta^*) 
	=& \left <-\nabla\mathcal{L}_\tau(\boldsymbol \beta^*) - \lambda_1\tilde {\boldsymbol \omega}_1 - \lambda_2\tilde {\boldsymbol \omega}_2, \hat{ \boldsymbol \beta} - \boldsymbol \beta^* \right >\\
	=& \left < \nabla\mathcal{L}_\tau(\boldsymbol \beta^*) , \boldsymbol \beta^* - \hat {\boldsymbol \beta} \right > 
	+ \lambda_1\left < \tilde {\boldsymbol \omega}_1, \boldsymbol \beta^* - \hat {\boldsymbol \beta} \right > \\
	&+ \lambda_2\left < \tilde {\boldsymbol \omega}_2, \boldsymbol \beta^* - \hat {\boldsymbol \beta} \right >.
	\end{aligned}
	\end{equation}
    Once again, for the sake of simplicity, we denote the right-hand-side terms at (\ref{threepart2}) as $I_1$, $I_2$, and $I_3$, respectively, and then turn our attention to their upper bounds.\\
{\tt(i)}	By Holder's inequality and Lemma \ref{l1cone}, we have
	\begin{equation}\label{three21}
	\begin{aligned}
	I_1 
	&\leq \|\nabla\mathcal{L}_\tau(\boldsymbol \beta^*)\|_\infty\|\hat {\boldsymbol \beta} - \boldsymbol \beta^* \|_1
	\leq \frac {\lambda_1} 2 \|\hat {\boldsymbol \beta} - \boldsymbol \beta^* \|_1\\
	&= \frac {\lambda_1} 2 \Big( \|(\hat {\boldsymbol \beta} - \boldsymbol \beta^*)_\mathcal{S} \|_1 + \|(\hat {\boldsymbol \beta} - \boldsymbol \beta^*)_{\mathcal{S}^c} \|_1\Big)\\
	&\leq \frac{2\lambda_1 (bd + 1)}{2bd + 1}\|(\hat {\boldsymbol \beta} - \boldsymbol \beta^*)_\mathcal{S} \|_1.
	\end{aligned}
	\end{equation}
{\tt(ii)} Using Holder's inequality and Lemma \ref{l1cone} again, we get
	\begin{equation}\label{three22}
	\begin{aligned}
	I_2 
	&\leq \lambda_1 \|\tilde {\boldsymbol \omega}_1\|_\infty\|\hat {\boldsymbol \beta} - \boldsymbol \beta^* \|_1
	\leq \lambda_1 \|\hat {\boldsymbol \beta} - \boldsymbol \beta^* \|_1\\
	&= \lambda_1 \Big( \|(\hat {\boldsymbol \beta} - \boldsymbol \beta^*)_\mathcal{S} \|_1 + \|(\hat {\boldsymbol \beta} - \boldsymbol \beta^*)_{\mathcal{S}^c} \|_1\Big)\\
	&\leq \frac{4\lambda_1 (bd + 1)}{2bd + 1}\|(\hat {\boldsymbol \beta} - \boldsymbol \beta^*)_\mathcal{S} \|_1.
	\end{aligned}
	\end{equation}
{\tt(iii)} For $I_3$, we have
	\begin{equation}\label{three23}
	\begin{aligned}
	I_3 
	&\leq \lambda_2 \|\tilde {\boldsymbol \omega}_2\|_\infty\|\hat {\boldsymbol \beta} - \boldsymbol \beta^* \|_1
	\leq \lambda_2d \|\hat {\boldsymbol \beta} - \boldsymbol \beta^* \|_1\\
	&= \lambda_2d \Big( \|(\hat {\boldsymbol \beta} - \boldsymbol \beta^*)_\mathcal{S} \|_1 + \|(\hat {\boldsymbol \beta} - \boldsymbol \beta^*)_{\mathcal{S}^c} \|_1\Big)\\
	&\leq \frac{4\lambda_2 d (bd + 1)}{2bd + 1}\|(\hat {\boldsymbol \beta} - \boldsymbol \beta^*)_\mathcal{S} \|_1.
	\end{aligned}
	\end{equation}
	
	Substituting (\ref{three21}), (\ref{three22}) and (\ref{three23}) into (\ref{threepart2}), it yields that
	\begin{equation}\label{D1}
	\begin{aligned}
	D_\mathcal{L}^s(\hat {\boldsymbol \beta}, \boldsymbol \beta^*) 
	\leq& \frac{2 \lambda_1 (2bd + 3)(bd + 1)}{2bd + 1}\|(\hat {\boldsymbol \beta} - \boldsymbol \beta^*)_\mathcal{S} \|_1\\
	\leq& \frac{2 \lambda_1 (2bd + 3)(bd + 1)}{2bd + 1}\sqrt{s}\|(\hat {\boldsymbol \beta} - \boldsymbol \beta^*)_\mathcal{S} \|_2,
	\end{aligned}
	\end{equation}
	where $s = \mid supp(\boldsymbol \beta^*)\mid$ is a sparsity parameter.
	
	In what follows, we employ Lemma \ref{RSC} to obtain a lower bound for the symmetric Bregman divergence. 
	At the first place, we denote $\hat {\boldsymbol \beta}_l := \boldsymbol \beta^* + l(\hat {\boldsymbol \beta} - \boldsymbol \beta^*), l \in (0,1]$ such that $\|\hat {\boldsymbol \beta}_l - \boldsymbol \beta^*\|_1 \leq r$ for some $r > 0$.
	
	In fact, if $\|\hat {\boldsymbol \beta} - \boldsymbol \beta^*\|_1 < r$, we can set $l = 1$ which means that $\hat {\boldsymbol \beta}_l = \hat {\boldsymbol \beta} $ and $\|\hat {\boldsymbol \beta}_l - \boldsymbol \beta^*\|_1 < r$;
	otherwise if $\|\hat {\boldsymbol \beta} - \boldsymbol \beta^*\|_1 \geq r$, we choose $l \in (0,1)$ such that $\|\hat {\boldsymbol \beta}_l - \boldsymbol \beta^*\|_1 = r$.
	
	Hence, $\hat {\boldsymbol \beta}_l$ falls into a local $\ell_1$ cone, i.e., $\hat {\boldsymbol \beta}_l \in \mathcal{C}(m, c_0, r)$. 
	Then by Lemma \ref{l1cone}, 
	\begin{equation}\label{betal}
	\begin{aligned}
	\|(\hat {\boldsymbol \beta}_l - \boldsymbol \beta^*)_{\mathcal{S}^c}\|_1 
	\leq& \frac{2bd + 3}{2bd + 1} \|(\hat {\boldsymbol \beta}_l - \boldsymbol \beta^*)_{\mathcal{S}}\|_1.
	\end{aligned}
	\end{equation}
	Then by Lemma \ref{RSC}, we have
	\begin{equation}\label{D2}
	D_\mathcal{L}^s(\hat {\boldsymbol \beta}_l , \boldsymbol \beta^*) \geq \frac{\kappa_{low}}{2}\|\hat {\boldsymbol \beta}_l - \boldsymbol \beta^*\|_2^2.
	\end{equation}
	By Lemma \ref{dsBd}, we have
	\begin{equation}\label{D3}
	D_\mathcal{L}^s(\hat {\boldsymbol \beta}_l , \boldsymbol \beta^*) \leq lD_\mathcal{L}^s(\hat {\boldsymbol \beta} , \boldsymbol \beta^*).
	\end{equation}
	Combining (\ref{D2}) and (\ref{D3}) wit (\ref{D1}), it yields that
	$$
	\|\hat {\boldsymbol \beta}_l - \boldsymbol \beta^*\|_2^2 
	\leq \frac{4 \lambda_1 (2bd + 3)(bd + 1)}{2bd + 1} \kappa_{low}^{-1} l\sqrt{s}\|\hat {\boldsymbol \beta} - \boldsymbol \beta^*\|_2.
	$$
	Because $\hat {\boldsymbol \beta} - \boldsymbol \beta^* = l^{-1}(\hat {\boldsymbol \beta}_l - \boldsymbol \beta^*)$, we have
	$$
	\begin{aligned}
	\|\hat {\boldsymbol \beta}_l - \boldsymbol \beta^*\|_2
	\leq&  \frac{4 \lambda_1 (2bd + 3)(bd + 1)}{2bd + 1} \kappa_{low}^{-1}\sqrt{sl}\\
	\leq&  \frac{4 \lambda_1 (2bd + 3)(bd + 1)}{2bd + 1} \kappa_{low}^{-1}\sqrt{s}.
	\end{aligned}
	$$
	Finally, by (\ref{betal}), we have
	$$
	\begin{aligned}
	\|\hat {\boldsymbol \beta}_l - \boldsymbol \beta^*\|_1 
	\leq& \frac{4(bd + 1)}{2bd + 1} \sqrt{s}\|(\hat {\boldsymbol \beta}_l - \boldsymbol \beta^*)\|_2\\
	\leq& \frac{16\lambda_1(bd + 1)^2(2bd + 3) }{(2bd + 1)^2}\kappa_{low}^{-1}s
	< r,
	\end{aligned}
	$$
	where the last inequality is from the assumption that $r \gtrsim \lambda_1\kappa_{low}^{-1}s$ and $n \geq c_3 m^2 t$ for a sufficiently large constant $c_3>0$.
	Because $\|\hat {\boldsymbol \beta}_l - \boldsymbol \beta^*\|_1 < r$, we have $\hat {\boldsymbol \beta}_l = \hat {\boldsymbol \beta}$, which means that
	\begin{equation}\label{pro1}
	\|\hat {\boldsymbol \beta} - \boldsymbol \beta^*\|_2
	\leq \frac{4 \lambda_1 (2bd + 3)(bd + 1)}{2bd + 1} \kappa_{low}^{-1}\sqrt{s} 
	\end{equation}
	holds with probability at least $1 - e^{-t}$.
	
	It remains to bound the probability so that the required condition $\|\nabla \mathcal{L}_\tau(\boldsymbol \beta^*)\|_\infty \leq \lambda_1/2$ in Lemma \ref{l1cone} holds.
	Following the argument used in the proof of  Sun et al. \cite[Theorem 1]{sun2020adaptive}, we take $\tau := \tau_0(n/t)^{1/(1+\delta)}$ for some $\tau_0 \geq \nu_{\delta}$
	and reach 
	$$
	P\{\|\nabla \mathcal{L}_\tau(\boldsymbol \beta^*)\|_\infty \geq 2\tau t/n\} \leq 2pe^{-t}.
	$$
	Hence, we have $\lambda_1/2 \geq 2\tau t/n$, that is $\lambda_1 \geq 4\tau_0(t/n)^{\delta/(1+\delta)}$. Then, together with (\ref{pro1}), we can prove that
	$$
	\|\hat{\boldsymbol \beta} - \boldsymbol \beta^*\|_2 
	\leq \lambda_1\kappa_{low}^{-1}\sqrt{s}
	$$
	holds with probability at least $1-(1+2p)e^{-t}$.
	\verb| |
\end{proof}

\end{appendices}

\bibliography{reference}
\bibliographystyle{mathphys}

\end{document}